\documentclass[pra,reprint,nofootinbib,amssymb,superscriptaddress]{revtex4-1}

\usepackage{natbib}
\usepackage{gensymb}
\usepackage{amsthm}
\usepackage{amsmath}
\usepackage{color}
\usepackage{bm}
\usepackage{amsmath,amssymb}
\usepackage{amsfonts}
\usepackage{dsfont}
\usepackage{graphicx}
\usepackage{float}
\usepackage{subfigure}
\usepackage{verbatim}
\usepackage{amsfonts}
\usepackage{bbold}
\usepackage{dcolumn}
\usepackage{bm}
\usepackage[dvipsnames]{xcolor}
\usepackage{listings}
\definecolor{blueprl}{RGB}{46,48,146}
\usepackage[pdftex,colorlinks=true,urlcolor=blueprl,citecolor=blueprl,linkcolor=blueprl]{hyperref}
\usepackage{mathrsfs}
\usepackage{filecontents}
\usepackage{mathtools}
\usepackage{times}
\usepackage{color}
\usepackage[dvipsnames]{xcolor}

\usepackage{braket}

\newtheorem{proposition}{Proposition}

\DeclareMathOperator{\sign}{sign}
\DeclareMathOperator{\sinc}{sinc}
\DeclareMathOperator{\arccosh}{arccosh}
\DeclareMathOperator{\rect}{rect}

\newcommand{\cam}[1]{\textcolor{violet}{ #1}}
\usepackage{soul}

\newcommand{\comn}[1]{\textbf{\textcolor{blue}{AM:#1~}}}

\date{\today}

\begin{document}

\title{Realizing an Unruh-DeWitt detector through electro-optic sampling of the electromagnetic vacuum}

\author{Sho Onoe}
\email{sho.onoe@uqconnect.edu.au}
\affiliation{Centre for Quantum Computation and Communication Technology, School of Mathematics and Physics, The University of Queensland, St. Lucia, Queensland, 4072, Australia}
\author{Thiago L. M. Guedes}
\affiliation{Department of Physics, University of Konstanz, D-78457 Konstanz, Germany}
\author{Andrey S. Moskalenko}
\email{moskalenko@kaist.ac.kr}
\affiliation{Department of Physics, KAIST, Daejeon 34141, Republic of Korea}
\author{Alfred Leitenstorfer}
\affiliation{Department of Physics, University of Konstanz, D-78457 Konstanz, Germany}
\author{Guido Burkard}
\affiliation{Department of Physics, University of Konstanz, D-78457 Konstanz, Germany}
\author{Timothy C. Ralph}
\affiliation{Centre for Quantum Computation and Communication Technology, School of Mathematics and Physics, The University of Queensland, St. Lucia, Queensland, 4072, Australia}

\date{\today}

\begin{abstract}
A new theoretical framework to describe the experimental advances in electro-optic detection of broadband quantum states, specifically the quantum vacuum, is devised. By making use of fundamental concepts from quantum field theory on spacetime metrics, the nonlinear interaction behind the electro-optic effect can be reformulated in terms of an Unruh-DeWitt detector coupled to a conjugate field during a very short time interval. When the coupling lasts for a time interval comparable to the oscillation periods of the detected field mode (i.e. the subcycle regime), virtual particles inhabiting the field vacuum are transferred to the detector in the form of real excitations. We demonstrate that this behavior can be rigorously translated to the scenario of electro-optic sampling of the quantum vacuum, in which the (spectrally filtered) probe works as an Unruh-DeWitt detector, with its interaction-generated photons arising from virtual particles inhabiting the electromagnetic vacuum. We discuss the specific working regime of such processes, and the consequences through characterization of the quantum light involved in the detection.
\end{abstract}

\maketitle

\section{Introduction}
Quantum field theory is one of the foundations of modern physics, answering many questions from the early days of quantum physics. Although it was initially seen as an approach restricted to particle and high-energy physics, it gradually became a fundamental working tool for branches of physics spanning from condensed matter and quantum optics to quantum-relativistic effects \cite{QFThistory}. The latter case involved the consideration of quantum fields on relativistic spacetime metrics, leading to novel implications and a (re)formulation of thermodynamical laws for systems in which both relativistic and quantum effects play major roles \cite{Bekenstein, Hawking_thermodynamics}. These studies led to the proposition of a range of extremely interesting effects, of which Unruh-Davies and Hawking radiation can be highlighted as some of the most intriguing. In the last two decades techniques from quantum information science \cite{QIhistory} have been merged in various ways with quantum field theory in curved spacetime, leading to investigations of information-related questions at the interface of quantum mechanics and relativity. Potential implications range from creating novel measurement protocols to answering fundamental questions about the quantization of gravity (the only fundamental force yet to be quantized) \cite{entanglementinLQG, Maldacena}. This new field is often referred to as relativistic quantum information.

The Unruh-Davies effect corresponds to the observation of thermal radiation in the quantum vacuum by an observer moving with constant proper acceleration in (Minkowski) spacetime \cite{Unruh, Davies}. It was initially proposed as an alternative observation possibility of the effect by which an inertial observer sees thermal radiation coming from the horizon of a black hole, the Hawking effect \cite{Hawkingeffect}. The mechanism behind both effects, which are linked by the equivalence principle, relies on the existence of horizons in spacetime \cite{Birrel}. Initial doubts about the observability of such radiation were allayed by the introduction of the Unruh-DeWitt (UDW) detector \cite{Unruh1984What}. This tool is the theoretical representation of a (usually) point-like monopole that couples to the quantum field of interest through a Hamiltonian (or Lagrangian) term linear in both monopole and field, with coupling strength and observer's worldline that can be varied at will {\cite{Schlicht_2004}}. The importance of this innovation was that the somewhat abstract particles seen in the relativistic field representations could really be mapped to the excitations of a simple, but convincing detector model.

One of the most interesting predictions based on such devices is the possibility of getting two space-like separated UDW detectors entangled through their coupling to the field even if the latter is in its ground (vacuum) state, i.e. the entanglement between the detectors appears before light-like particles can travel the distance between the detectors \cite{Valentini, Reznik}. This discovery was given the name vacuum entanglement, since it is believed that the vacuum works as a reservoir of entanglement for quantum systems. Although similar interaction Hamiltonians can be found in several branches of physics, one of the most characteristic examples of a real-world realization of such interactions is the light-matter coupling.

Two-level atomic or artificial-atom systems 
are potential physical embodiments of the theoretical UDW detector \cite{Unruh_Scully}. Some ground-breaking experiments have managed to make considerable advances in this regard, by either effectively controlling the (slowly-varying) time-dependent coupling between the system and a resonator \cite{Cleland} or another artificial atom \cite{Schoelkopf}, or even by rapidly switching on such coupling through subcycle activation of electronic quantum wells in an optical cavity \cite{Alfred_wells}. Theory based on the latter experiment predicts emission of virtual cavity photons as a consequence of the nonadiabatic change of the ground state of the system, a feature that closely resembles the emission of Unruh-Davies particles \cite{Bastard, Auer, Yablonovitch}.

Related studies employing a subcycle-light probe to electro-optically sample low-frequency (quantum) electric fields have recently led to the measurement of the (broadband) electric-field variance of the vacuum \cite{Alfred_vacuum, Faist}. These results have raised a debate on whether the measurements deliver real estimations for vacuum fluctuations or their outcomes represent a by-product of squeezing \cite{Faist}. We remark that theoretical support to the former interpretation has already been provided \cite{Andrey, Buhmann}. In fact, the spectra obtained from such measurements should closely resemble the spectra of Unruh-like particles detected by a finite-lifetime observer moving with constant proper acceleration \cite{Thiago, RovelliandMartinetti}. Moreover, these works provide a possible generalization of the known relation between two-mode squeezing and Unruh-Davies radiation \cite{Leonhardt_book, Birrel}, and reinforce the idea of electro-optic sampling being an analogue of a vacuum measurement in which the observer is found in a non-inertial reference frame \cite{Matthias}. 

Nevertheless, the connection between the quantum optical process of squeezing and these types of relativistic effects is subtle. As an example, consider the dynamical Casimir effect in which a rapidly oscillating mirror can produce photons from the vacuum \cite{Wilson2011}. As it is the physical acceleration of the mirror that produces the photons, there is a clear connection between these photons and the Unruh-Davies radiation. On the other hand, the oscillating mirror affects the vacuum in a very similar way to an oscillating boundary condition produced by a harmonic-optical-pump-induced rapidly varying refractive index in a nonlinear crystal. The latter is a well-known optical method of producing photons from the vacuum via squeezing. Although the physical mechanisms are distinct, the overall effect is to produce a very similar interaction with the vacuum.

The present article aims at strengthening the connections between effects predicted in quantum field theory in curved spacetime \cite{Birrel} and the recent measurements of vacuum fluctuations based on the electro-optic effect. To this end, we make use of theoretical tools pertaining to both fields and propose a way to tie these perspectives together.
The thereby developed approach allows for rigorous differentiation between different regimes in electro-optic sampling, a measurement scheme in which the probe pulse is subcycle with respect to the sampled-frequency field. In particular, we identify the regime in which the ultrafast switching on and off of the interaction, controlled by a strong coherent probe field, directly maps virtual particles from the vacuum into real excitations of the probe field. We achieve this identification by firstly studying the behavior of a simple harmonic UDW detector coupled to a subcycle field mode (cf. Fig. \ref{FigVisualUnruh}). We then place this alongside a novel analysis of the quantum electro-optic sampling and thereby find the regime in which the actions of the (UDW-)detector-field and the (electro-optic) field-field interactions are approximately equivalent. In fact, this equivalence defines a new and experimentally feasible optical variant of the UDW detector \cite{Unruh_Gooding}. We further discuss the presence of thermality in the predicted vacuum quadrature moments. This thermality arises as a consequence of entanglement breakage, thus providing us with direction towards the possible observation and harnessing of entanglement from the vacuum through subcycle electro-optic sampling techniques.

The remainder of this paper is organized as follows: after discussing field modes and particularly subcycle modes in the next two sections, we proceed to our analysis of subcycle sampling using a UDW in section IV. After a brief overview of the techniques involved in connecting the UDW detection to electro-optic sampling in section V, we move in section VI to making this connection rigorous. Our conclusions and outlook are then presented in the final section.

\begin{figure}[t]
\includegraphics[width=0.45\textwidth]{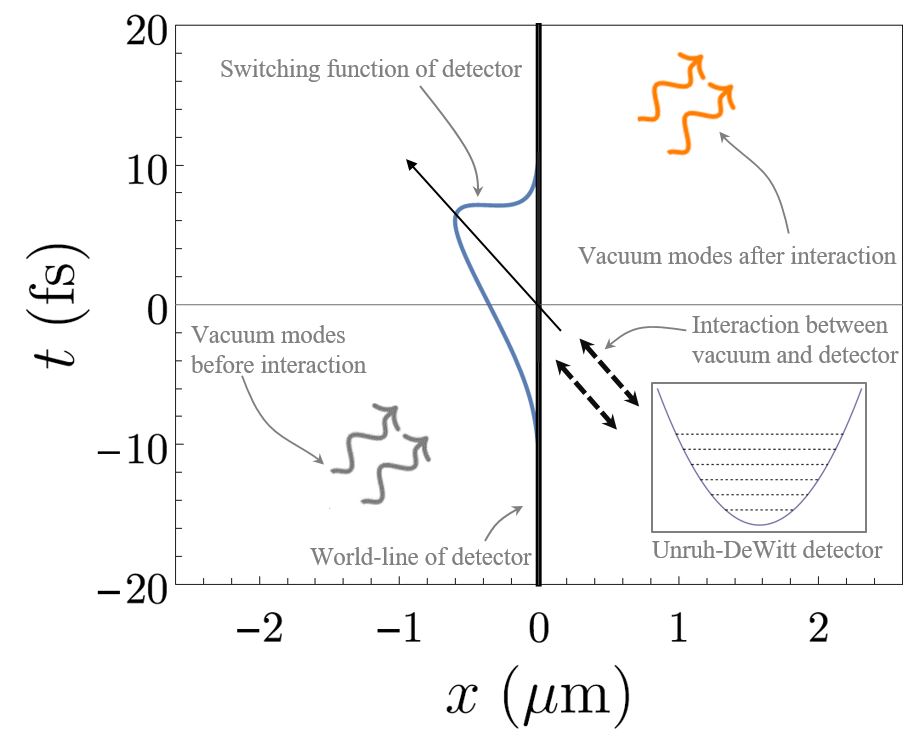}
\caption{A schematic representation of the Unruh-DeWitt (UDW) detector. The inset at the bottom right shows the harmonic oscillator of the {UDW} detector, with dashed lines representing the different energy levels with constant gap $\hbar \omega_u$. The vacuum modes of the field are represented by the gray curly lines in the bottom-left corner. The UDW detector follows the trajectory shown via the bold black line. It interacts with the vacuum during a very short time interval {defined by} the switching function (represented by the blue line). In the Schrödinger picture, both the UDW detector and the vacuum states are affected by this interaction. The orange curvy lines represent the evolved modes of the field that are no longer in the initial vacuum state.
}
\label{FigVisualUnruh}
\end{figure}

\section{Preliminaries}
\subsection{Scalar and conjugate fields}
We start our discussion by introducing the basic tools of this work. In the second quantization formalism, a classical field and its conjugate are promoted to field operators by translating the Poisson bracket between them into a commutation relation. For the specific case of scalar bosonic fields, both the field and its conjugate are derived from the Klein-Gordon Lagrangian density and are given as solutions of the Klein-Gordon equation. These solutions, when decomposed into single-frequency modes, allow for a description of the field operators in terms of an infinite series of quantum harmonic oscillators \cite{Peskin}. When considering fields that propagate exclusively (or predominantly) along a given direction, a decomposition of these operators in terms of travelling-wave modes turns out to be particularly useful \cite{Blow1990,Vogel}, since the fields so described behave essentially as fields in (1+1) dimensions (i.e. fields depending on one space and one time coordinate), while keeping some (3+1)-dimensional features, namely the cross-sectional area. Therefore, their action is given by a fourfold integral over the Lagragian density. In terms of these modes, the (right-moving) field operator is given as
\begin{equation}
\begin{gathered}
\hat{\Phi}(t,x) =
\int_{-\infty}^{\infty}\hspace{-3.5 mm} \mathrm{d} \omega \, \Phi_{\omega}(t,x) \hat{a}_{\omega} \, ,
\\
\Phi_{\omega}(t,x) =\sqrt{\frac{\hbar c}{4 \pi n_{\omega} |\omega| A}}e^{-i \omega(t- \frac{n_\omega x}{c})}
 \, ,
\end{gathered}
\end{equation}
where we have neglected the left-moving modes for simplicity. This simplification is justified because we only consider {interactions} with the right-moving modes. Here and in the following expressions, $A$ is a normalization parameter with units of area and accounts for the transversal extension of the field. The parameter $n_{\omega}=n_{|\omega|}$ is included to allow propagation velocities other than the speed of light, i.e. $v_\omega=c/n_{\omega}$ ($n_{\omega}$ will assume the role of refractive index once we connect these concepts to quantum-optical systems). For free fields, $n_{\omega}=1$. $\hbar$ is the reduced Planck constant.

The corresponding conjugate field in terms of the travelling-wave mode decomposition is given by: 
\begin{equation}
\begin{gathered}
\hat{\Pi}(t,x)=
\int_{-\infty}^{\infty}\hspace{-3.5 mm} \mathrm{d} \omega \, \Pi_{\omega}(t,x) \hat{a}_{\omega} \, ,
\\
\Pi_{\omega}(t,x) = -\sign(\omega)i\sqrt{\frac{\hbar|\omega|}{4 \pi n_{\omega}{c} A}}e^{-i \omega(t- \frac{n_\omega x}{c})}
 \, .
\end{gathered}\label{EqnConjufateField}
\end{equation}
In this article we have adopted the convention $\hat{a}_{{-\omega}}=\hat{a}_{{\omega}}^{\dag}$. This justifies the use of the terminology positive-frequency modes for annihilation operators and negative-frequency modes for creation operators. The corresponding commutation relation is then $[\hat{a}_{\omega},\hat{a}_{\omega'}^{\dag}]=\; \delta(\omega-\omega')\sign(\omega)$, consistent with the relation $[\hat{\Phi}(t,x),\hat{\Pi}(t, x')]\approx i(\hbar c/n^2_\omega A)\delta(x-x')$ dictated by the correspondence principle.

\subsection{Discrete decomposition of scalar and conjugate fields}
Following the description of \cite{Leaver1986SpecDecomp, Law2000Continuous, Rohde_2007, Onoe2019}, we introduce a complete orthonormal set of discrete (nonmonochromatic) bosonic operators $\{\hat{a}_{i},\hat{a}_{j},...\}$:
\begin{equation}
\hat{a}_{i} = \int_{-\infty}^{\infty}\hspace{-3.5 mm}\mathrm{d}\omega \, f_i(\omega) \hat{a}_{\omega}\, .
\label{nonmonochromatic}
\end{equation}
This set of operators satisfies the commutation relations $[\hat{a}_{i},\hat{a}_{j}^{\dag}]=\delta_{ij}$ and $[\hat{a}_{i},\hat{a}_{j}]=0$.
The scalar and conjugate-field operators can be expanded in terms of the operators in this discrete basis set with the aid of the decomposition \cite{Onoe2019}
\begin{equation}
\hat{a}_\omega=\sum_i([\hat{a}_{\omega},\hat{a}_i^{\dag}]\hat{a}_i+[\hat{a}_i,\hat{a}_{\omega}]\hat{a}_i^{\dag})\label{decompositionforladder} \, ,
\end{equation}
allowing these fields to be cast in the forms
\begin{gather}
\hat{\Phi}(t,x)= \sum_{i}\Phi_{i}(t,x)\hat{a}_{i}+h.c. \, , \label{phi_discrete}
\\
\hat{\Pi}(t,x)= \sum_{i}\Pi_{i}(t,x)\hat{a}_{i}+h.c. \, , \label{pi_discrete}
\end{gather}
respectively. $\Phi_{i}(t,x)$ ($\Pi_{i}(t,x)$) can be interpreted as the scalar (conjugate) field mode that is annihilated by the operator $\hat{a}_{i}$. We introduce $\Phi_{i}(\omega)$ and $\Pi_{i}(\omega)$ as the Fourier transforms of these discrete modes, 
\begin{equation}
\begin{gathered}
\Phi_{i}(t,x) = \int_{-\infty}^{\infty}\hspace{-3.5 mm}\mathrm{d}\omega'\, e^{-i \omega'(t- \frac{n_\omega x}{c})} \Phi_i(\omega') \, ,
\\
\Pi_{i}(t,x) = \int_{-\infty}^{\infty}\hspace{-3.5 mm}\mathrm{d}\omega'\, e^{-i \omega'(t- \frac{n_\omega x}{c})} \Pi_i(\omega') \, .
\end{gathered}
\end{equation}
For a desired $\Phi_i$ mode profile, the spectral-decomposition coefficients  for $\hat{a}_{i}$ in Eq.~\eqref{nonmonochromatic} can be written with respect to the Fourier transform as follows:
\begin{equation}
\begin{aligned}
f_i(\omega)&=\sign(\omega) \sqrt{\frac{4 \pi n_\omega |\omega|A}{\hbar c}}\Phi^* _{i}(\omega)
\\
&=i\sqrt{\frac{4 \pi n_\omega c A}{\hbar |\omega|}}\Pi^*_{i}(\omega)   \, .
\label{f_of_ai}
\end{aligned}
\end{equation}

\section{Subcycle modes}\label{SecSubcycleModes}
In this section we shall consider for simplicity a free Klein-Gordon scalar field ($n_{\omega}=1$) in its vacuum state (i.e. $\hat{a}_{\omega}\ket{0}=0, \; \forall \; \omega > 0$). With such assumptions in mind, we briefly describe the properties of nonmonochromatic modes, as given by Eqs. \eqref{nonmonochromatic}-\eqref{f_of_ai}, when one of these modes is tailored to match a Gaussian pulse-like profile with very short time extension (i.e. the subcycle regime). A mode is considered to enter the subcycle regime when the envelope decays at a time interval shorter than its own inverse central frequency.
\subsection{Gaussian profile}
We introduce a normalized Gaussian-profile mode of the form
\begin{gather}
\Phi_g(t,x)=\frac{1}{(2\pi)^{1/4}}\sqrt{\frac{\hbar \sigma c}{A \omega_0}}e^{-\sigma^2(t-\frac{x}{c}-t_0)^2-i (t-\frac{x}{c}) \omega_0 }\, ,
\label{Gaussian}
\end{gather}
where we have chosen a particular mode-index $i=g$ from the infinitely many possible values of $i$. This Gaussian pulse has a central (or carrier) frequency $\omega_0/2\pi$ and a temporal variance of {$1/(2\sigma^2)$, with the position of its amplitude maximum crossing the point $x=0$ at time} $t=t_0$. {We consider the Gaussian-profile} mode to be in the subcycle regime when $2\pi/\omega_0>\sqrt{8}/\sigma$, i.e. one period of the carrier frequency is at least as long as four standard deviations of the Gaussian envelope.

The operator that annihilates the Gaussian mode, $\Phi_g(t,x)$, is given by:
\begin{gather}
\hat{a}_g = \int_{-\infty}^{\infty}\hspace{-3.5 mm} \mathrm{d}\omega \,  f_g(\omega)\hat{a}_{\omega} \, ,  \label{EqnDefoffPOmega_a}
\\
f_g(\omega) = \frac{1}{(2\pi)^{1/4}} \sign(\omega) \sqrt{\frac{|\omega|}{ \omega_0 \sigma}}e^{-i t_0(\omega-\omega_0)-\frac{(\omega-\omega_0)^2}{4 \sigma^2}} \, 
\label{EqnDefoffPOmega}
\end{gather}
(without the loss of generality, we shall consider $t_0=0$ throughout the calculations for simplicity).
This operator is a normalized bosonic operator satisfying the commutation relation $[\hat{a}_g,\hat{a}_g^{\dag}]=1$. It can be further decomposed into positive- and negative-frequency components, according to:
\begin{gather}
\hat{a}_g = \cosh(\theta_{g})\hat{a}_g^{(+)}+\sinh(\theta_{g})\hat{a}_g^{(-)}{}^{\dag} \, .
\label{+and-decomposition}
\end{gather}
Each of these terms is defined as follows:
\begin{gather}
\hat{a}_g^{(+)}= \frac{1}{\cosh(\theta_{g})}\int_{0}^{\infty}\hspace{-3.5 mm}\mathrm{d}\omega \, f_g(\omega)\hat{a}_{\omega} \label{ladder+}\, ,
\\
\hat{a}_g^{(-)}= \frac{1}{\sinh(\theta_{g})}\int_{0}^{\infty}\hspace{-3.5 mm} \mathrm{d}\omega\, f^*_g(-\omega)\hat{a}_{\omega} \, ,
\\
\theta_{g} = \arccosh\left(\int_0^{\infty}\hspace{-3.5 mm} \mathrm{d}\omega\, |f_g(\omega)|^2 \right) \, . \label{theta_g}
\end{gather}
The operators $\hat{a}_g^{(\pm)}$ are also normalized bosonic operators, i.e. $[\hat{a}_g^{(\pm)},\hat{a}_g^{(\pm)}{}^{\dag}]=1$. We note that in general $\theta_{g}$ is very small and therefore the mode is dominated by its positive frequency component (i.e. $\hat{a}_g \approx \hat{a}_g^{(+)}$). The negative frequency component, $\hat{a}_g^{(-)}$, becomes significant in cases such as the subcycle regime. We highlight that $\hat{a}_g^{(\pm)}$ are generally not orthogonal to each other, i.e.
\begin{equation}
{} [\hat{a}_g^{(+)},\hat{a}_g^{(-)}{}^{\dag}]=\frac{1}{\cosh(\theta_g)\sinh(\theta_g)}\int_{0}^{\infty}\hspace{-3.5 mm}\mathrm{d}\omega \, f_g(\omega)f_g(-\omega) \, .
\label{nonorthogonality}
\end{equation}
In view of Eq. \eqref{nonorthogonality}, a completely orthogonal decomposition of Eq. \eqref{+and-decomposition} might be preferred; this is achieved by further decomposing $\hat{a}_g^{(-)}$ into an orthogonal and a parallel component with respect to $\hat{a}_g^{(+)}$ \cite{Rohde_2007, Onoe2019}, i.e.
\begin{equation}
\hat{a}_{g}^{(-)}= \cos(\theta_{g,\perp}) \hat{a}_{g,\perp}^{(-)}+\sin(\theta_{g,\perp})e^{-i \phi_{g,\perp}} \hat{a}_{g}^{(+)} \, , \label{orthogonalize}
\end{equation}
where we have defined the following components:
\begin{gather}
\hat{a}_{g,\perp}^{(-)}=
\frac{\hat{a}_{g}^{(-)}-\big[\hat{a}_g^{(-)},\hat{a}_g^{(+)}{}^{\dag}\big] \hat{a}_g^{(+)}}{\sqrt{1-\Big|\big[\hat{a}_g^{(-)},\hat{a}_g^{(+)}{}^{\dag}\big]\Big|^2}} \, ,
\\
\theta_{g,\perp}= \arcsin \left(\Big|\big[\hat{a}_g^{(+)},\hat{a}_g^{(-)}{}^{\dag}\big]\Big|\right) \, , \label{EqnDefofthetagperp}
\\
\phi_{g,\perp}=\text{Arg}\left([\hat{a}_g^{(+)},\hat{a}_g^{(-)}{}^{\dag}]\right) \, .
\end{gather}
Utilizing these results, Eq. \eqref{+and-decomposition} can be recast in the following manner:
\begin{equation}
\begin{aligned}
\hat{a}_g=\cosh(\theta_g)\hat{a}_g^{(+)}&+\sinh(\theta_g)\sin(\theta_{g,\perp})e^{i \phi_{g,\perp}} \hat{a}_g^{(+)}{}^{\dag}
\\
&+\sinh(\theta_{g})\cos(\theta_{g,\perp}) \hat{a}_{g,\perp}^{(-)}{}^{\dag} \, .
\end{aligned}\label{EqnSpectralDecompositionSimplified}
\end{equation}
If $\theta_{g,\perp}=0$, then $\hat{a}_g$ is in a two-mode squeezed state between $\hat{a}_g^{(+)}$ and $\hat{a}_{g,\perp}^{(-)}$. If $\theta_{g,\perp}=\pi/2$, then $\hat{a}_g$ is in a single-mode squeezed state of $\hat{a}_g^{(+)}$. When $0 <\theta_{g,\perp} < \pi/2$ the state has a combination of {single- and two-mode} squeezing terms.

\subsection{Properties of a subcycle mode}
\label{SecSubcycleProp}
In this section, we analyse some phase-space properties of the subcycle mode described by $\hat{a}_g$ \cite[Chap. 4]{serafini2017quantum}.
The vacuum-state projection on mode $g$, $\hat{\rho}_g (\hat{a}_g,\hat{a}^\dagger_g) =\text{tr}_{\perp g}\{|0\rangle \langle 0|\}$ (where $\text{tr}_{{\perp g}}$ stands for trace over the complementary subspace, i.e. all modes orthogonal to $g$)
 is described by a Gaussian quasiprobability distribution in phase space that can be fully characterized by the first and second moments of the quadrature operator $\hat{X}_{g}(\phi)= \hat{a}_g e^{-i \phi}+\hat{a}_g^{\dag}e^{i \phi}$ \cite{weedbrook2012gaussian, adesso2014continuous}.
Its first moment is $\braket{\hat{X}_{g}(\phi)}=\text{tr}\{\ket{0}\bra{0} \hat{X}_g(\phi)\}=\text{tr}_g\{\hat{\rho}_g \hat{X}_{g}(\phi)\}=0$, while the second moment {reads}:
\begin{equation}
\braket{\hat{X}_{g}^2(\phi)}=1+2\braket{\hat{a}_g^{\dag}\hat{a}_g}+2\Re[{\braket{\hat{a}_g^2}}e^{- 2 i \phi}] \, ,\label{EqnSubcycleVariance1}
\end{equation}
\begin{equation}
\braket{\hat{a}_g^{\dag}\hat{a}_g}=\sinh^2(\theta_{g}) \, , \label{EqnSubcycleVariance2}
\end{equation}
\begin{equation}
\braket{\hat{a}_g^2}=\cosh(\theta_{g})\sinh(\theta_{g})\sin(\theta_{g,\perp}) \exp[i \phi_{g,\perp}] \, .
\label{EqnSubcycleVariance3}
\end{equation}

\begin{figure}[t]
\includegraphics[width=0.45\textwidth]{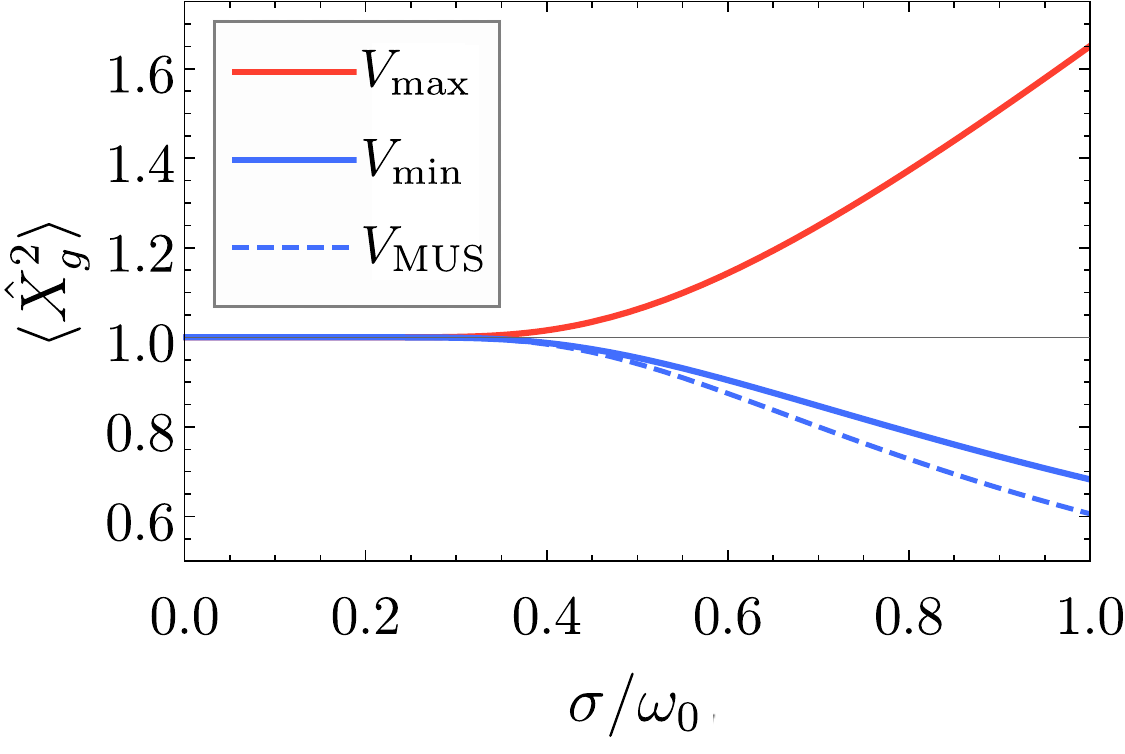}
\caption{Numerical plot of the $Q$- and $P$-quadrature variances for the subcycle mode $g$, Eq. \eqref{EqnDefoffPOmega_a}, in dependence of the normalized inverse time extension of the Gaussian profile, $\sigma/\omega_0$. The top solid line (red) represents the {$P$-quadrature} variance, while the bottom solid line (blue) corresponds to the {$Q$-quadrature} variance. The dashed line represents the {$Q$-quadrature} variance one would obtain by assuming that both variances would characterize a single-mode (minimum-uncertainty) squeezed state (MUS), i.e. {$\langle \hat{Q}^2_g(\sigma)\rangle_{\text{MUS}}  = 1/\langle \hat{P}^2_g (\sigma)\rangle$}. The difference between the dashed and the bottom solid line can be attributed to extra thermal photons within $\hat{a}_g$
. 
It can be seen from the graph that the {$Q$- and $P$-quadrature} variances start to deviate from quantum shot noise (i.e. variance of $1$) at around $\sigma \approx 0.45 \omega_0$, when the Gaussian-profile mode enters the subcycle regime.}
\label{FigPlotSubcycle}
\end{figure}
As can be seen from Eq. \eqref{EqnSubcycleVariance2}, the parameter $\theta_{g}$, as defined by Eq. \eqref{theta_g}, characterizes the amount of particles that can be found in the mode $g$, $\braket{\hat{N}_g}=\braket{\hat{a}_g^{\dag}\hat{a}_g}$. Since no actual particles are supposed to be found in the vacuum state (it is, after all, the ground sate of all the quantum harmonic oscillators that compose the field), any nonzero expectation value of the particle number operator, $\braket{\hat{N}_g}$, can be attributed to virtual particles. Heisenberg's uncertainty principle predicts that, for sufficiently short time intervals, energy fluctuations allow such particles to briefly come into existence, being annihilated in particle-antiparticle collisions right after (please note that the photon is its own antiparticle). Considering that $\braket{\hat{N}_g}$ tends to zero as the time extension of the Gaussian profile, Eq. \eqref{Gaussian}, increases, it is consistent to attribute the non-zero particle number to the subcycle character of the mode $g$, with increasing particle numbers as the time scale of the mode becomes shorter.

The parameter $\theta_{g,\perp}$ is defined by Eq. \eqref{EqnDefofthetagperp}, and can be understood as the amount of overlap between positive- and negative-frequency terms within $\hat{a}_g$. This parameter can be found in Eq. \eqref{EqnSubcycleVariance3} and parametrizes the amount of quadrature correlation (through the covariance) within this mode. The quadrature correlation between two orthogonal quadratures, ${X}_g(\phi)$ and $\hat{X}_g(\phi+\pi/2)$, is defined as
\begin{equation}
\begin{aligned}
\frac{1}{2}\left\langle\left\{\hat{X}_g(\phi),\hat{X}_g(\phi+\pi/2)\right\}\right\rangle &-\langle\hat{X}_g(\phi)\rangle \langle\hat{X}_g(\phi+\pi/2)\rangle
\\&=2\mathrm{Im}[\langle\hat{a}_g^2\rangle e^{-2i\phi}]\, .
\end{aligned}
\end{equation}
We then define the quadrature correlation within the mode described by $\hat{a}_g$ to be the maximum of $2\mathrm{Im}[\langle\hat{a}_g^2\rangle e^{-2i\phi}]$ with respect to the phase of its argument,
\begin{equation}
2|\langle\hat{a}_g^2\rangle|= 2|\cosh(\theta_{g})\sinh(\theta_{g})\sin(\theta_{g,\perp})|^2\, .
\end{equation}
When $\theta_{g,\perp}=0$, there are no quadrature correlations, in the sense that the difference between the variances of any pair of orthogonal quadratures, $\langle \hat{X}^2_g(\phi)-\hat{X}^2_g(\phi+\pi/2)\rangle = 4 \Re \{\langle \hat{a}_g^2 \rangle e^{-2i\phi}\}\propto\sin (\theta_{g,\perp}) $, is zero; in this case, the (virtual) particles of $\hat{a}_g$ obey a thermal distribution. On the other hand, when $\theta_{g,\perp}=\pi/2$, there is maximal quadrature correlation, meaning that the virtual particles of $\hat{a}_g$ are distributed according to a pure single-mode squeezed state [see Eq. \eqref{EqnSpectralDecompositionSimplified}].

{For a Gaussian pulse, both the average photon-count/variance, parametrized by $\theta_g$, and the difference between the maximum and minimum variances, parametrized by $\theta_{g,\perp}$, increase as the temporal features of $\hat{a}_g$ become more subcycle (in other words, they increase alongside $\sigma$). To show this result, we plot the maximal and minimal (Minkowski) vacuum-state variances of the quadratures in Fig. \ref{FigPlotSubcycle} as functions of $\sigma$. The maximal and minimal variances are given by the variances of $\hat{P}_g={\hat{X}_g(\pi/2)}$ and $\hat{Q}_g={\hat{X}_g(0)}$, respectively, (i.e. $V_{\mathrm{max}}(\sigma)=\braket{\hat{P}_g^2(\sigma)}$ and $V_{\mathrm{min}}(\sigma)=\braket{\hat{Q}_g^2(\sigma)}$). It can be seen from this plot that $V_{\mathrm{max}}$ and $V_{\mathrm{min}}$ become larger and smaller, respectively, as the temporal extension of the Gaussian pulse decreases. Thus {a subcycle vacuum mode} has more virtual particles and quadrature correlation with a shorter temporal profile. It is noted that $\hat{Q}_g$ and $\hat{P}_g$ are related via a $-\pi/2$ phase shift applied on $\hat{a}_g$. This differs from a $-\pi/2$-phase shift applied directly on the single frequency annihilation operators (i.e. $\hat{a}_\omega\rightarrow {\hat{a}_{\omega}}e^{-i\pi/2 \sign(\omega)}$), as $\hat{a}_g$ consists of both positive- and negative-frequency modes.

We note that, in general, $\theta_g$ and $\theta_{g,\perp}$ can be varied independently of each other if we do not restrict ourselves to a Gaussian waveform. For example, a Rindler mode (a field mode as seen by an observer that moves with constant proper acceleration, i.e. along a hyperbolic worldline confined to a [Rindler] wedge of spacetime) is a special case for which $\theta_{g,\perp}=0$, with the right (left) Rindler mode operator (related to the right [left] Rindler wedge) being given by a Bogoliubov transformation, namely the two-mode squeezing, between two modes defined in Minkowski spacetime. This means that the detected particles in this (Rindler) mode follow thermal statistics, having correlation/entanglement only between creation and annihilation operators of orthogonal modes. For a Gaussian mode we see that the amount of thermalization (i.e. the difference between the blue line and blue dashed line in Fig. ~\ref{FigPlotSubcycle}) is very small and hence most of the correlation lies within this same mode.
\section{Subcycle sampling with an Unruh-DeWitt detector}\label{SecSubCycleUnruh}
For a number of interesting phenomena encountered in (curved-space) quantum field theory, realistic measurement schemes are either intractable or impossible to implement in view of the extreme and unusual conditions. One of the main detection models proposed to overcome these theoretical difficulties, the UDW detector, relies on a simple (dipole-like) system interacting with the field during a properly selected time interval while following a trajectory given by a desired worldline. Such construction, albeit essentially theoretical, finds parallel in simple light-matter interacting systems such as qubits or electrons in quantum wells \cite{Cleland, Schoelkopf, Alfred_wells, Bastard}; however, a fully tailorable ultrashort time-dependent coupling to the electromagnetic field (with both switching on and off under control) has, to the best of our knowledge, not yet been fully implemented in the lab. Aiming at overcoming these limitations, we introduce here the description of a bosonic UDW which will later be properly translated in terms of an optical system of easy implementation in the lab. 

\subsection{Unitary operator for the evolution of the Unruh-DeWitt detector}
One can reproduce (up to minor deviations related to differences between bosonic and the usually employed fermionic excitations \cite{Hotta2015Partner}) an idealized UDW detector with an energy gap of $\hbar\omega_u$, by considering a quantum system with a one-dimensional harmonic-oscillator behavior.
The detector's annihilation operator, $\hat{u}$, satisfies the commutation relation $[\hat{u},\hat{u}^{\dag}]=1$. We consider the initial state of the UDW detector to be the ground-state, i.e. $\hat{u}\ket{0_u}=0$.

We let this UDW detector interact with the right-moving Klein-Gordon conjugate field. The interaction Hamiltonian is defined as follows:
\begin{gather}
\hat{H}_I(\tau)= A \lambda(\tau)\hat{Q}(\tau) \hat{\Pi}\big(t(\tau),x(\tau)\big)\, ,\label{EqnInterHamUDW}
\\
\hat{Q}=\sqrt{\frac{\hbar {c}}{2  A}}\left(\hat{u} e^{-i {\omega_u} \tau}+\hat{u}^{\dag} e^{i {\omega_u} \tau}\right) \, .
\label{UDWinteraction}
\end{gather}
In these equations, the dependence on the transversal extension of the field, $A$, cancels out when integration over the transversal coordinates is taken into account in the derivation of the interaction Hamiltonian. It is worth noting, however, that by keeping $A$ explicitly in Eq. \eqref{UDWinteraction} one allows for its interpretation as the cross-sectional area of the detector, i.e. travelling-wave modes with transversal extension given by $A$ are detected. Here and in what follows, $\tau$ is the proper time parametrizing the world-line $\big(t(\tau),x(\tau)\big)$ of the UDW detector. The coupling strength (also known as switching function), $\lambda(\tau)$, characterizes the proper-time interval {during which the UDW} detector is coupled to the field.

The corresponding interaction unitary operator, $\hat{U}_{I,\mathcal{T}}$, is given by
\begin{equation}
\hat{U}_{I,\mathcal{T}}=\mathcal{T} \exp\Big[\frac{-i}{\hbar}\int_{-\infty}^{\infty} \hspace{-3.5 mm}\mathrm{d}\tau \,\hat{H}_I(\tau)\Big]\, ,
\end{equation}
where $\mathcal{T}$ stands for time-ordering protocol. {By choosing an operational regime in which time-ordering effects can be neglected \cite{Silberhorn, Quesada2014, Lipfert2018}, we justify the usage of $\hat{U}_I$, the equivalent of $\hat{U}_{I,\mathcal{T}}$ in the absence of $\mathcal{T}$, as the interaction unitary operator in the remainder of this work.} Fig. \ref{FigVisualUnruh} illustrates the coupling between the conjugate field and the UDW detector.

\subsection{Sampling {of a} Gaussian{-}profile mode with an Unruh-DeWitt detector}

{When one considers a stationary UDW detector following the world-line $(\tau,0)$ and interacting with the (conjugate) field in its vacuum state for all time with $\lambda = \text{const}$, the detector's response is solely related to the field mode described by the operator $\hat{a}_{\omega_u}$. As the coupling strength, $\lambda(\tau)$, acquires some time dependence, the resulting modulation of the interaction in time leads to a frequency broadening over the probed mode(s), i.e. to a modulation of the envelope of a measured-mode profile with carrier frequency $\omega_u$. The more localized in proper time this coupling is (meaning that the detector-field coupling switches on and off within a finite time interval), the broader the frequency band of the field that the detector probes. We can therefore utilize the switching function to control the time interval and frequency band for which a field is probed. The resultant probed field mode has its envelope dictated by the switching function and a carrier frequency $\omega_u$.

The UDW detector can interact with a mode with a Gaussian profile by setting the switching function to be
\begin{equation}
\lambda(t)=\eta \exp\big(-\sigma_u^2(t-t_u)^2\big)\label{EqnSwitchingFunction} \, .
\end{equation}
$\eta$ can be interpreted as the coupling{-strength amplitude for the interaction} between the detector and the conjugate field. {The parameters $\sigma_u$ and $t_u$ define, respectively, the inverse temporal extension (i.e. bandwidth) and the initial time shift of the pulse-like coupling.} By setting the switching function and the energy {gap of the UDW} detector to be consistent with {Eq. \eqref{Gaussian} (i.e. $\sigma_u=\sigma$, $t_u=t_0$ and $\omega_u=\omega_0$),} the interaction unitary operator can be simplified to (refer to App. \ref{AppInputOutput})
\begin{align}
 \hat{U}_{I}&= \exp [\theta_{u} (\hat{a}_g\hat{u}^{\dag}-h.c.)]\,. \label{EqnDeWittUnitary}
\end{align}
$\theta_{u} = -\frac{\eta}{2}\sqrt{\frac{\omega_0}{\sigma}}\left(\frac{\pi}{2}\right)^{1/4}$ can be interpreted as the effective interaction strength between $\hat{a}_g$ and $\hat{u}$. This unitary operator resembles a beam-splitter{-}type interaction between the {UDW} detector, {described by the operator} $\hat{u}$, and the Gaussian{-}profile mode {corresponding to} $\hat{a}_g$. {In the subcycle regime when $\langle \hat{N}_g \rangle \neq 0  $,} this unitary operator maps the virtual particles in the vacuum field to real excitations in the UDW detector.
\subsection{Detection mechanism}
The Heisenberg evolution of the UDW-detector's annihilation operator as dictated by $\hat{U}_I$ leads to an exchange of particles between the two modes involved in Eq. \eqref{EqnDeWittUnitary}, i.e.
\begin{gather}
\hat{u}'=\hat{U}_I^{\dag}\hat{u}\hat{U}_I=\cos(\theta_{u})\hat{u}+\sin(\theta_{u})\hat{a}_g \, .
\end{gather}
Analogously to the discussion preceeding Eq. \eqref{EqnSubcycleVariance1}, the Gaussian character of the phase-space distributions for the state $\text{tr}_{\perp u'}\{|0_u\rangle|0\rangle \langle 0|\langle 0_u|\}$ guarantees that a full description of this state can be done solely in terms of the expectation values of $\hat{X}_{u'}(\phi)= {\hat{u}'}e^{-i \phi}+{\hat{u}'}{}^{\dag}e^{i \phi}$ and corresponding variances, $\braket{\hat{X}^l_{{u'}}(\phi)}=\text{tr}\{|0_u\rangle|0\rangle \langle 0|\langle 0_u|\hat{X}^l_{u'}(\phi)\}$ with $l=1,2$. These read
\begin{gather}
\braket{\hat{X}_{u'}}=\cos(\theta_u)\braket{\hat{X}_{u}}+\sin(\theta_u)\braket{\hat{X}_{g}} \, ,
\label{EqnUnruhDeWittXu}
\\
\braket{\hat{X}_{u'}^2}=\cos^2(\theta_u) \braket{\hat{X}_{u}^2}+\sin^2(\theta_u) \braket{\hat{X}_{g}^2} \, .
\label{EqnUnruhDeWittXu^2}
\end{gather}
{Due to the particle exchange between the detector and field induced by the evolution, Eq. \eqref{EqnDeWittUnitary}, the variance of the detector's quadrature operator contains $\hat{X}_g$ terms.}
The extra factor of $\sin(\theta_{u})$ ($\sin^2(\theta_{u})$) is related to the efficiency of the detector.

We find that the first moment is $\braket{\hat{X}_{{u'}}}=0$, while substituting the results from Eqs. (\ref{EqnSpectralDecompositionSimplified}) and (\ref{EqnSubcycleVariance1}) gives for the second moment
\begin{equation}
\begin{aligned}
\braket{\hat{X}_{{u'}}^2(\phi)}=1&+2\sin^2(\theta_{u})\braket{\hat{a}_g^{\dag}\hat{a}_{g}}
\\&+2\sin^2(\theta_{u}) \Re[\braket{\hat{a}_g^2}e^{- 2 i \phi}]\, .
\end{aligned}
\label{EqnDetectorVariance}
\end{equation}
{When $\sin(\theta_{u})=1$, the Unruh-DeWitt detector has unit efficiency and {maps all} of the virtual particles in $\hat{a}_g$ to excitations of $\hat{u}$, giving an identical result to Eq. (\ref{EqnSubcycleVariance1}).
\subsection{Quantum-optical analogue}
A UDW detector with a coupling strength of subcycle character can detect virtual particles in the field vacuum, $|0\rangle$, which otherwise would remain concealed to inertial observers. The complete subcycle switching on and off, however, has yet to be implemented in a realizable system \cite{Alfred_wells, Bastard}. In what follows, we shall show that a closely related optical system can in fact provide the desired coupling.

As a prototype to our UDW detector we study the broadband $\chi^{(2)}$ interaction in a nonlinear crystal \cite[Chap. 16]{Yariv1989} driven by a strong coherent pump, $\hat{E}(\omega) \rightarrow \alpha_g(\omega)$ [refer to Eq. (\ref{EqnSchi2Compact1}) for comparison]. The corresponding evolution operator for the states propagating through the crystal contains in its exponent all possible bilinear combinations of annihilation/creation operators with frequency variables running continuously over the ranges determined by a frequency-dependent interaction strength. In order to achieve a similar structure for the interaction between the UDW detector and conjugate field, we recast Eq. \eqref{EqnDeWittUnitary} in the form
\begin{gather}
\hat{U}_I= \exp [\hat{S}^{(+)}_{I} + \hat{S}_I^{(-)}] \, , \label{EqnQuantumOpticsInterpretation}
\\
\hat{S}^{(+)}_{I}=  \int_0^{\infty} \hspace{-3.5 mm}\mathrm{d}\omega \, {\zeta_g(\omega)}\alpha_g(\omega_u-\omega)\hat{a}_{\omega}\hat{u}^{\dag}-h.c. \, ,
\\
\hat{S}^{(-)}_{I}= \int_0^{\infty}\hspace{-3.5 mm} \mathrm{d}\omega \, {\zeta_g(\omega)}\alpha_g(\omega_u+\omega)\hat{a}_{\omega}^{\dag}\hat{u}^{\dag} -h.c. \, ,	
\end{gather}
where we have defined the following quantities:
\begin{gather}
\alpha_g (\omega)= e^{-\frac{\omega^2}{4\sigma^2}+i\omega t_0}  \label{alpha_gomega}
 \, ,
\\
\zeta_g (\omega)= \sign(\omega)\frac{\eta}{2\sigma} \sqrt{\frac{|\omega|}{2}} \, .
\label{analogyparameters}
\end{gather}
The action term $\hat{S}_I^{(+)}$ has the shape of a beam-splitter-type interaction, while $\hat{S}_I^{(-)}$ corresponds to a squeezing-type interaction. Note that while $\hat{a}_g$ contains both positive- and negative-frequency operators, $\hat{S}_I^{(\pm)}$ contains solely operators for which $\omega > 0$.

Tracing a parallel to the broadband $\chi^{(2)}$ interaction, $\hat{S}^{(+)}_{I}$ ($\hat{S}^{(-)}_{I}$) resembles the action of sum-(difference-)frequency generation. It is noted that the Gaussian form of Eq. \eqref{alpha_gomega} comes from the profile of the switching function, while for the boradband $\chi^{(2)}$ interaction, this parameter comes from the Gaussian profile of the strong coherent pump. As a result, the strong coherent pulse in the nonlinear interaction replaces the role of the switching function in the UDW detector. $\zeta_g(\omega)$ is a frequency-dependent function, reflecting the degree of fulfillment of the phase-matching condition in an analogue nonlinear-optical setup \cite[Chapter 2.3]{Boyd}, \cite{Dorfman2016}, taken here in the limit case of perfect phase matching (when we may effectively neglect the frequency dependence of $n_{\omega}$ for the given length of the utilized nonlinear crystal). The main difference between the evolution defined by Eqs. \eqref{EqnQuantumOpticsInterpretation}-\eqref{analogyparameters} and the actual broadband $\chi^{(2)}$ interaction is the restriction of one of the frequencies (and therefore one of the annihilation/creation operators) to the fixed value $\omega_u$ defined by the UDW energy gap in the former. This represents, however, no shortcoming in this analogy, since a similar restriction can also be achieved through frequency post-selection of the outgoing photons. We shall therefore further analyse such an effective nonlinear interaction from the perspective of a UDW-detector implementation.
\section{Bridging quantum optics and relativistic quantum information}
Before giving the precise description of the actual nonlinear interaction that reproduces a UDW-detector behavior, we shall briefly introduce and justify the approximations playing a major role in the following sections.
The action of sum- and difference-frequency generation, when we treat the nonlinear crystal to be driven by a classical coherent pump, can generally be written in the form
\begin{gather}
\hat{S}= \int_{-\infty}^{\infty}\hspace{-3.5mm}\mathrm{d}\omega\int_{-\infty}^{\infty}\hspace{-3.5mm}  \mathrm{d}\omega'\, S(\omega,\omega') \hat{a}_{\omega}^{\dag} \hat{u}_{\omega'} \, , \label{EqnS2General}
\\
\hat{U}= \exp[\hat{S}]\label{EqnS2General0}\, ,
\end{gather}
{where the operators $\hat{u}_{\omega}=\hat{u}^\dagger_{-\omega}$ can in general be related to a new set of bosonic modes (therefore commuting with $\hat{a}_{\omega}$ for all frequencies) or be the same as $\hat{a}_{\omega}$ depending on the structure of the susceptibility tensor of the nonlinear crystal, $\chi^{(2)}_{ijk}(\omega+\omega',\omega,\omega')$. In this section we consider the prior, satisfying the commutation relation $[\hat{u}_{\omega},\hat{u}^\dagger_{\omega'}]=\sign(\omega)\delta(\omega-\omega')$ and $[\hat{u}_{\omega},\hat{a}^\dagger_{\omega'}]=0$. Since Eqs. \eqref{EqnS2General}-\eqref{EqnS2General0} define a unitary operator, the condition $S(\omega,\omega')=-S^*(-\omega,-\omega')$ must be satisfied.} 

Heisenberg evolution of the annihilation operators according to this unitary operator (i.e. $\hat{u}_{\omega}'=\hat{U}^{\dag}\hat{u}_{\omega}\hat{U}$) leads to operators described by an infinite series of convolutions over increasing numbers of frequencies \cite{Thiago}. A more transparent and simpler, yet fully analytical, presentation of the core properties (e.g. whether the interaction can be modelled as a beam-splitter-type interaction or squeezing-type interaction) of the system under study can be achieved through introduction of a technique we shall refer to as the first-order unitary evolution.

\subsection{Discrete-mode decomposition}
By decomposing $\hat{u}_{\omega}$ in terms of an arbitrary discrete basis set $\{\hat{u}_i,\hat{u}_i^{\dag}\}$ (refer to Eq. \eqref{decompositionforladder}), Eq. (\ref{EqnS2General}) can be recast in the form:
\begin{gather}
\hat{S} = \sum_{i} \theta^{(1)}_{i}( \bar{a}_{i}\hat{u}_{i}^{\dag}-h.c.)\, ,\label{prescription1}
\\
\bar{a}_{i} = \frac{1}{\theta^{(1)}_{i}}[\hat{u}_{i},\hat{S}]\, ,
\\
\theta^{(1)}_{i} = \sqrt{\Big|\big[[\hat{u}_{i},\hat{S}],[\hat{S},\hat{u}_{i}^{\dag}]\big]\Big|}\label{prescription3} \, .
\end{gather}
$\theta^{(1)}_{i}$ is the normalization factor ensuring $\bar{a}_i$ is normalized: $|[\bar{a}_i,\bar{a}_i^{\dag}]|=1$. When $[\bar{a}_i, \bar{a}_i^{\dag}]=1$, $\bar{a}_i$ follows the properties of an annihilation operator and hence $\bar{a}_i=\hat{a}_i$. When $[\bar{a}_i, \bar{a}_i^{\dag}]=-1$, $\bar{a}_i$ follows the properties of a creation operator and hence $\bar{a}_i=\hat{a}_i^{\dag}$. Note that in general  $[\bar{a}_i,\bar{a}_{j}^{\dag}] \neq 0$ for $i\neq j$. As these operators are not orthogonal, a closed (nonperturbative) expression for the Baker-Hausdorff lemma is not possible, and it motivates us to introduce the first-order unitary approximation.
\subsection{First-order unitary evolution}
The $n$th-order unitary evolution is a technique we develop in App. \ref{AppNthOrderUni}. It simplifies the evolution of an operator by considering a slightly modified, easier to handle, action such that the evolution of operators exactly satisfies the Baker-Hausdorff lemma (expressed in terms of the original action) to a desired $n$th order.  For the first-order unitary evolution, the evolution of the annihilation operators for a specific mode $g$ is computed via a unitary evolution for which only the $g$-containing terms in Eq. \eqref{prescription1} are considered, i.e. the evolution operator is approximated by:
\begin{gather}
\hat{U}^{[1]}=\exp[ \hat{S}^{[1]}] \, ,
\\
\hat{S}^{[1]}= \theta^{(1)}_{{g}}(\hat{u}_{{g}}^{\dag}\bar{a}_{g}-h.c.)\, .\label{EqnFirstorderunitaryaction}
\end{gather}

The first-order unitary evolution is then found to be 
\begin{align}
\hat{u}_{g}^{[1]} &= \hat{U}^{[1]}{}^{\dag}\hat{u}_{g}\hat{U}^{[1]}\label{EqnFirstOrderUnitaryEval1}
\\&=\begin{cases}
\cos\big(\theta^{(1)}_{{g}}\big)\hat{u}_g-\sin\big(\theta^{(1)}_{{g}}\big)\hat{a}_g,\, & \text{if} \; [\bar{a}_g,\bar{a}_{g}^{\dag}]=1\cam{\,;}
\\ \cosh\big(\theta^{(1)}_{{g}}\big)\hat{u}_g-\sinh\big(\theta^{(1)}_{{g}}\big)\hat{a}_g^{\dag},\, & \text{if} \; [\bar{a}_g,\bar{a}_g^{\dag}]=-1\cam{\;.}
\end{cases}\nonumber
\end{align}
Since the approximation is made for the action occurring in the exponent (i.e. $\hat{S}\approx \hat{S}^{[1]})$, the first-order unitary evolution contains, in terms of the Baker-Hausdorff lemma expansion, terms of arbitrarily high order in the interaction strength, $\theta_{g}^{(1)}$. In fact, Eq. \eqref{EqnFirstOrderUnitaryEval1} represents a nonperturbative result. The first-order approximation might seem to be inconsistent with our goal of calculating both first and second moments of the quadrature operators, since the latter is quadratic in $\theta^{(1)}_{g}$. 
We show in App. \ref{AppNumerics}, however, that a complete analysis of the evolved annihilation and creation operators in Eq. \eqref{EqnFirstOrderUnitaryEval1} in terms of the second-order unitary approximation leads to additional terms with negligible contribution to the moments we are interested in.
These terms are only formally needed to guarantee that the commutation relations for the operators are accurate to second order after the evolution. We will therefore not consider them in the following discussions and derivations, focusing instead on the first-order evolution, as described by Eq. \eqref{EqnFirstOrderUnitaryEval1}.

\section{Electro-optic sampling}
We are now ready to properly address the real-world counterpart of our UDW detector. Assuming that different polarization components of the vector fields can be treated independently, we consider the components of the electric field (the conjugate of the vector potential in quantum  electrodynamics) to be proportional to conjugate Klein-Gordon fields\footnote{In the Lorenz gauge, the (massless) Klein-Gordon and the Helmholz equations (descibing, respectively, the Klein-Gordon and vector-potential fields) have the same functional form. Apart from belonging to different representations of the Lorentz group, leading to different behaviors under Lorentz transformations (see, e.g. \cite[pp. 81-88]{Birrel}), each component of the vector potential behaves as an independent Klein-Gordon field, therefore justifying the approximation.}. 
We match each of the two polarization components of the field, $\nu \in \{s,z\}$ -- cf. Fig.~\ref{FigExperimentalSetup}, to an independent massless Klein-Gordon conjugate field, as defined by Eq. \eqref{EqnConjufateField}, through the relation:
\begin{gather}
\begin{aligned}
\hat{E}_\nu(t,x) &=
\int_{-\infty}^{\infty}\hspace{-3.5 mm} \mathrm{d} \omega \, E_{\omega.\nu}(t,x) \hat{a}_{\omega} \, 
\\
&=-\sqrt{\frac{1}{\varepsilon_0}} \hat{\Pi}_\nu(t,x)\; ,
\end{aligned}
\\
E_{\omega,\nu}(t,x) = \text{sign}(\omega)i\sqrt{\frac{\hslash|\omega|}{4 \pi n_{\omega}{c} \varepsilon_0 A}}e^{-i \omega(t- \frac{n_\omega x}{c})} \; .
\end{gather}
From now on, $n_{\omega}$ will stand for the refractive index of the nonlinear crystal in which the {$\chi^{(2)}$} process takes place (we assume it to be isotropic in terms of the linear optical properties). We have included into the field amplitudes the constant factor of $-1/\sqrt{\varepsilon_0}$, with ${\varepsilon_0}$ being the vacuum permittivity, so that the commutation relations between creation and annihilation operators are preserved: $[\hat{a}_{\omega,\nu},\hat{a}_{\omega',\nu'}^{\dag}]=\; \delta(\omega-\omega')\sign(\omega)\delta^{\nu}_{\nu'}$ (in terms of the discrete modes described by Eqs. \eqref{nonmonochromatic}-\eqref{f_of_ai}, one has $[\hat{a}_{i,\nu},\hat{a}_{j,\nu'}^{\dag}]=\; \delta^i_{j}\delta^{\nu}_{\nu'}$).
\subsection{Hamiltonian of a $\chi^{(2)}$ electric-field interaction}\label{SecInterHamNonLinear}

The electro-optic effect utilizes the optical response of the polarization in a nonlinear crystal to mediate an effective nonlinear interaction of the electric field with itself. In the specific case of the Pockels effect \cite{Boyd},  the outgoing (i.e. generated) field depends quadratically on the incoming field, with a proportionality constant dependent on the effective second-order susceptibility of the medium, $\chi^{(2)}$, determined by its corresponding tensor components. This interaction is one of the possible mechanisms behind squeezing \cite[Chapter 8.2.3]{Vogel}. It can also be used to probe an electric field of given polarization and frequency range by additionally impinging a copropagating strong coherent probe field having suitable polarization and frequency range. In the latter case, the probe polarization is dictated by the structure of the susceptibility tensor, while its frequency range should be chosen to have minimal overlap with the sampled field's spectrum. The duration of the probe pulse should be shorter than the period or the characteristic time scale of the sampled radiation, giving the subcycle resolution required for the electro-optic sampling. Quantum versions of experiments of this kind have managed to detect the electric-field variance of the electromagnetic vacuum in the mid-infrared (MIR) \cite{Alfred_vacuum}  and terahertz \cite{Faist} frequency ranges.
We shall therefore focus on a similar measurement scheme.

We represent the sampled-frequency photons by $\hat{a}^\dagger_{\Omega}$ ($\Omega$ being a MIR frequency), while the detected-frequency-range particles are represented by $\hat{a}^\dagger_{\omega}$ ($\omega$ being a near-infrared, or NIR, frequency).  The nonlinear interaction is modulated by a short and strong coherent electric-field pulse with an amplitude given by $\alpha$. We assume that this pulse is short enough so that the interaction time, dictated by $\alpha_p(t)$, is subcycle relative to the sampled MIR  field. As in Ref.~\cite{Andrey}, we consider the incoming probe field to be linearly polarized along the $z$-direction (see Fig.~\ref{FigExperimentalSetup}), while the sampled vacuum modes are restricted to the perpendicular $s$-polarization, so as is the newly generated (through the Pockels effect) quantum correction to the NIR field (this restriction is enforced by the structure of the second-order susceptibility tensor of zincblende-type materials).

We model the resulting effective interaction through the following Hamiltonian:
\begin{equation}
\hat{H}_{\chi}(t)= \int_{-\infty}^{\infty}\hspace{-3.5 mm} \mathrm{d}x \,\lambda \hat{E}_z(t,x)
\hat{E}_s(t,x) \hat{E}_s(t,x)\rect\left(\frac{x}{L}\right)\,.
\end{equation}
The coefficient $\lambda = \frac{A \varepsilon_0 d}{2}$ includes the cross-sectional area {$A$}, the $s$-polarized-field permutation factor of $1/2$ (avoids double counting)  and the coupling constant $d=-n^4 r_{41}$, expressed in terms of the electro-optic (susceptibility) coefficient, $r_{41}$, and the refractive index of the crystal,  $n=n_{\omega_p}$, at the central frequency of the probe, $\omega_p$. rect$(x/L)$ is a rectangular distribution with value 1 for $-L/2 \leq x \leq L/2$
 and 0 otherwise, representing the spatial extension of the nonlinear crystal of thickness $L$.

We consider the $z$-polarized field to be strongly displaced via $\hat{D}=\exp[\alpha \hat{a}_{p,z}^{\dag}-\alpha^* \hat{a}_{p,z}]$ (with $|\alpha| \gg 1$), leading to a strong (semi-)classical coherent field. The annihilation operator $\hat{a}_{p,z}$ is defined through Eqs. \eqref{nonmonochromatic} and \eqref{f_of_ai} with a (complex) electric-field waveform
\begin{equation}
E_p (\omega) = \frac{-1}{\sqrt{4 \pi}N_p}\left(\frac{1}{2 \pi}\right)^{1/4}\sqrt{\frac{|\omega|\epsilon_0}{n_{\omega_p} \sigma_p}}e^{-\frac{(\omega-\omega_p)^2}{4 \sigma_p^2}+i (\omega t_p+	\phi_p)} \, , \label{waveformE}
\end{equation}
where $N_p$ is a normalization constant ensuring $[\hat{a}_p,\hat{a}_p^{\dag}]=1$.
As the $z$-polarized field is in a strong coherent state, we may utilize the mean-field approximation to express the corresponding operator in terms of its coherent amplitude, leading to:
\begin{equation}\label{EqnEOHamiltonainMeanField}
\hat{H}_{\chi}(t)= \int_{-\infty}^{\infty} \hspace{-3.5 mm}\mathrm{d}x \,\lambda \alpha_p(t,x)
\hat{E}(t,x) \hat{E}(t,x)\rect\left(\frac{x}{L}\right)\, .
\end{equation}
The polarization subscripts will from this point forward be omitted as we treat the electric field  with $z$-polarization (semi-)classically and the remaining field operators have the same polarization. $\alpha_p (t,x) = E_p(t,x)\alpha + E^*_p(t,x)\alpha^*$ is the real probe  amplitude for a given (complex) pulse profile $E_p(t,x)$ [see {, e.g.,} Eq. \eqref{Gaussian}].

The exponent of the evolution operator for such interaction is defined as \begin{equation} \label{EqnActionEO}
\hat{S}=-\frac{i}{\hbar} \int_{-\infty}^{\infty}\hspace{-3.5 mm}\mathrm{d}t\, \hat{H}_{\chi}(t) \, .
\end{equation}
Assuming negligible overlap between frequencies in the NIR and MIR, we split the electric field into the detector-frequency range and the sampled-frequency range,
\begin{equation}
\hat{E}(t,x)=\hat{E}_{\Omega}(t,x)+\hat{E}_{\omega}(t,x) \, ,
\end{equation}
where the relation $[\hat{a}_{\omega},\hat{a}^\dagger_{\Omega}]=0$ is satisfied due to the difference in their frequency ranges. We neglect $E_p(t,x)\hat{E}_\omega(t,x)\hat{E}_\omega(t,x)$ and $E_p(t,x)\hat{E}_\Omega(t,x)\hat{E}_\Omega(t,x)$ as they are highly oscillatory terms, averaging out to contributions close to zero through the rotating-wave approximation. We then integrate with respect to space and time to obtain (cf. App. \ref{AppNonLinHam})
\begin{equation}
\hat{S}=\int_{|\Omega|<\Lambda}\hspace{-3.5 mm} \mathrm{d}\Omega \int_{\Lambda < |\omega|}\hspace{-3.5 mm}\mathrm{d}\omega \, S(\Omega,\omega)\hat{a}_{\Omega}\hat{a}_{\omega}^{\dag}  \, ,\label{EqnSchi2Compact1}
\end{equation}
where
\begin{gather}
S(\Omega,\omega)= \alpha_p(\omega-\Omega) \label{what_is_S}
\zeta_{\Omega,\omega} \, ,
\\ \label{EqnPhaseMatching1}
\zeta_{\Omega,\omega} = -i\sign(\omega \Omega)\frac{\lambda L}{ A c \epsilon_0 }\sqrt{\frac{|\omega \Omega|}{n_{\omega} n_{\Omega}}}  \sinc\left(\eta_{\omega,\Omega}\right) \, ,
\\
{\eta_{\Omega,\omega}=\frac{L}{2c}\left[\omega(n_{\omega}-n_{\omega-\Omega})-\Omega (n_\Omega-n_{\omega-\Omega})\right]\, .}\label{EqnPhaseMatching2}
\end{gather}
We note that $\hat{S}^\dagger=-\hat{S}$ is fulfilled as $\alpha_p (\omega) = E_p(\omega)\alpha + E_p^*(-\omega)\alpha^*$ satisfies $\alpha_p(\omega)=\alpha_p^*(-\omega)$. We have introduced a transition frequency, $\Lambda$, in order to avoid frequency crossing between the MIR and the NIR. $\zeta_{\omega,\Omega}$ determines the phase matching between $\hat{a}_{\omega}$ and $\hat{a}_{\Omega}$ and sinc$(x)=\sin (x)/x$.
\subsection{First-order unitary evolution in a nonlinear crystal}

In this subsection we implement the first-order unitary evolution to reduce Eq. (\ref{EqnSchi2Compact1}) to a form resembling Eq. (\ref{EqnQuantumOpticsInterpretation}). To do this, we must introduce a bosonic mode operator that resembles $\hat{u}$, i.e. a mode to which we can assign the role of a (UDW) detector.
The UDW detector is described by a harmonic oscillator with a well-defined energy gap of $\hbar \omega_u$. 
We introduce a narrow-frequency-band mode centered at $\widetilde{\omega}$ to mimic the UDW detector's (discrete) single-frequency mode:
\begin{equation}
\hat{\mathfrak{u}}_{\widetilde{\omega}}=\int_{-\infty}^{\infty} \hspace{-3.5 mm}\mathrm{d}\omega \,\frac{1}{\sqrt{\Delta \omega}}\text{rect}\left(\frac{\widetilde{\omega}-\omega}{\Delta\omega}\right)\hat{a}_{\omega} \, . \label{filter}
\end{equation}
The bandwidth $\Delta \omega$ can be as small as band-pass filters allow in real-world experiments. On top of such restrictions, decreasing $\Delta \omega$  will also filter out more photons, meaning less photons will be detected in the output. While one can span the whole NIR frequency range in terms of modes of the form \eqref{filter} with non-overlapping frequency windows of width $\Delta \omega$, we shall focus on a single such frequency window. By taking $\hat{u}_g\rightarrow \hat{u}_{\widetilde{\omega}}$ in Eq. \eqref{EqnFirstorderunitaryaction}, one can see that the mode operator corresponding to $\bar{a}_g$ has the form
\begin{gather}
\bar{\mathfrak{a}}_{\widetilde{\omega}} = \frac{1}{\theta^{(1)}_{\widetilde{\omega}}} [\hat{S},\hat{\mathfrak{u}}_{\widetilde{\omega}}^{\dag}] =\int_{|\Omega|<\Lambda}\hspace{-3.5 mm} \mathrm{d}\Omega\, f_{\widetilde{\omega}}(\Omega)\hat{a}_{\Omega},
\\
f_{\widetilde{\omega}}(\Omega)=\frac{1}{\theta_{\widetilde{\omega}}^{(1)}\sqrt{\Delta \omega}}\int_{\widetilde{\omega}-\Delta \omega/2}^{\widetilde{\omega}+\Delta \omega/2} \hspace{-3.5 mm}\mathrm{d}\omega\, \alpha_p(\omega-\Omega) \zeta_{\Omega,\omega} \; ,\label{waveformint}
\end{gather}
where
$\theta^{(1)}_{\widetilde{\omega}}= \sqrt{|[[\hat{S},\hat{\mathfrak{u}}_{\widetilde{\omega}}^{\dag}],[\hat{\mathfrak{u}}_{\widetilde{\omega}},\hat{S}]]|}$. In the limit of sufficiently small $\Delta\omega$,
$f_{\widetilde{\omega}}(\Omega)\approx  \frac{\sqrt{\Delta \omega}}{\theta_{\widetilde{\omega}}^{(1)}} \alpha_p(\widetilde{\omega}-\Omega) \zeta_{\Omega,\widetilde{\omega}}$. One can see that it closely approximates the respective expression one would expect for the actual UDW case:
$f_g(\Omega)=\frac{1}{\theta_u}\alpha_g(\omega_u-\Omega)\zeta_g(\Omega)$.
The main source of discrepancy between $f_{\widetilde{\omega}}(\omega)$ and $f_g(\Omega)$ can be attributed to the phase matching parameters, $\zeta_{\Omega,\tilde{\omega}}$ and $\zeta_{g}$. As a result, in the case of perfect phase matching, $\alpha_p(\widetilde{\omega}-\Omega)$ takes the role of $\alpha_g(\omega_u-\Omega)$. In this regime, the envelope of the probed field mode is determined by {$|E_{p}(t,0)|$}, while the carrier frequency is determined by $|\widetilde{\omega}-\omega_p|$. In Fig. \ref{FigWaveFormElectroOptic}, we plot the temporal waveforms of $f_{\widetilde{\omega}}(\Omega)$ and $E_{p}(t,0)$ for a given choice of $\widetilde\omega$. It is found that the temporal width of the former is slightly larger than that of $E_{p}(t,0)$ due to the difference in ($\widetilde\omega$-dependent) phase velocity. 
\begin{figure}[t]
\includegraphics[width=0.45\textwidth]{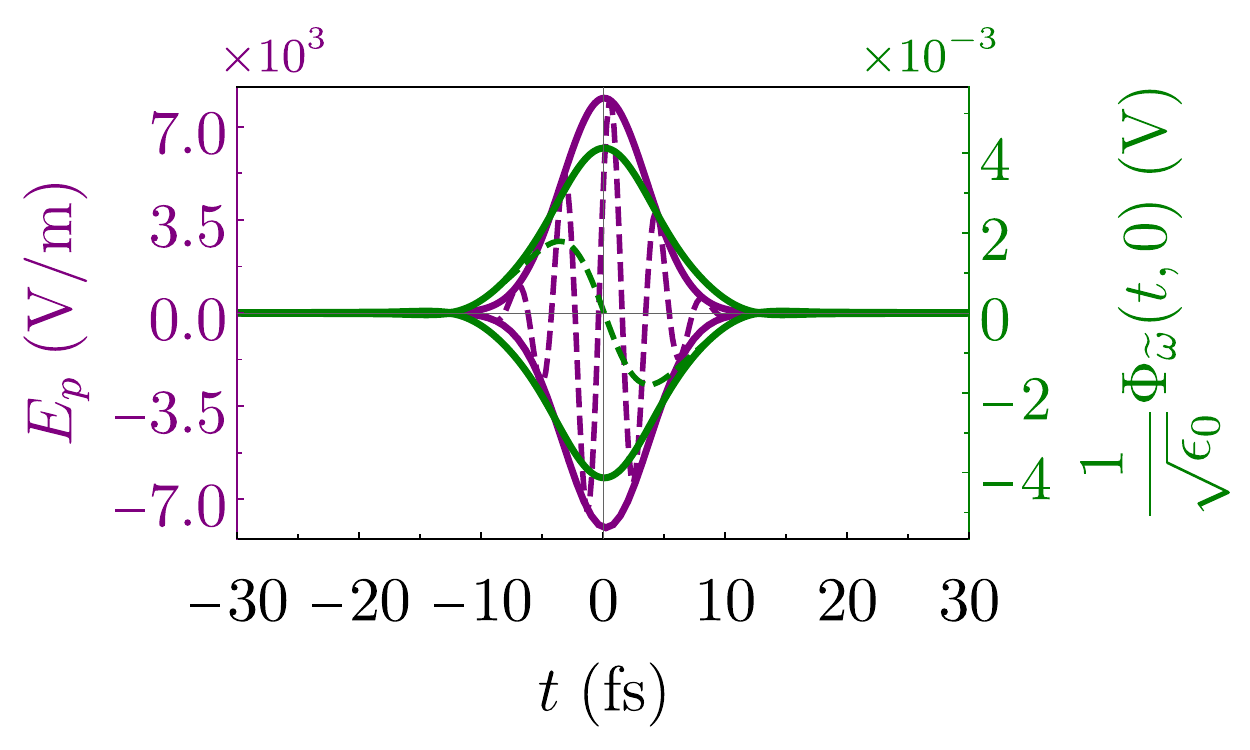}
\caption{(a) Amplitude of the scalar-field mode for the probed waveform $f_{\widetilde{\omega}}(\Omega)$, as given by Eqs. \eqref{phi_discrete}-\eqref{f_of_ai}. The green solid lines represent the envelope of the probed subcycle mode, $\pm\big|\sqrt{\frac{1}{\varepsilon_0}}\Phi_{\widetilde\omega}(t,0)\big|$. Its real part is given by the dashed green line. (b) The purple solid lines show $\pm|E_p(t,0)|$, the envelope of the probe pulse that drives the interaction, with a dashed purple line representing its real part. We have set {$t_p=0$ and} $\widetilde{\omega}=\omega_p+1.5\sigma_p$, where {$\omega_p/(2\pi)=255$ ~THz}, $\sigma_p=\sqrt{2 {\log 2}}/t_{1/2}$ and {$t_{1/2}=5.8$ ~fs}.}
\label{FigWaveFormElectroOptic}
\end{figure}

We compute the first-order unitary evolution on the filtered NIR-frequency operator as follows:
\begin{gather}
\hat{\mathfrak{u}}_{{\widetilde{\omega}}}'\approx \hat{\mathfrak{u}}_{{\widetilde{\omega}}}^{[1]}= \hat{U}^{[1]}{}^{\dag}\hat{\mathfrak{u}}_{{\widetilde{\omega}}}\hat{U}^{[1]} \, , \label{crazy_a}
\\
\hat{U}^{[1]}= \exp[\theta^{(1)}_{\widetilde{\omega}}(\bar{\mathfrak{a}}_{\widetilde{\omega}}\hat{\mathfrak{u}}^\dagger_{\widetilde{\omega}}-h.c.)]\, .
\end{gather}
This is a simple two-mode interaction between $\bar{\mathfrak{a}}_{\widetilde{\omega}}$ and $\hat{\mathfrak{u}}_{\widetilde{\omega}}$. If $\bar{\mathfrak{a}}_{\widetilde{\omega}}=\hat{\mathfrak{a}}_{\widetilde{\omega}}$, the interaction between these two modes is a beam-splitter-type interaction described by Eq. \eqref{EqnDeWittUnitary}. In this regime, the Hamiltonian of the nonlinear electric-field interaction can be modelled as an interaction between a UDW detector $\hat{\mathfrak{u}}_{\widetilde{\omega}}$ and the subcycle field mode $\hat{\mathfrak{a}}_{\widetilde{\omega}}$. This association is not possible in case $\bar{\mathfrak{a}}_{\widetilde{\omega}}=\hat{\mathfrak{a}}_{\widetilde{\omega}}^{\dag}$, for which the interaction is modelled as a two-mode squeezing interaction. In other words, when measurements are conducted in the regimes for which $\bar{\mathfrak{a}}_{\widetilde{\omega}}=\hat{\mathfrak{a}}_{\widetilde{\omega}}$ holds, the non-linear interaction promotes the (vacuum) virtual-particles to real excitations in the NIR frequencies.

In real-world experiments, the detection of the mode corresponding to Eq. \eqref{filter} can be approximated by photon counting with setups in which the detected photons are restricted to a given (narrow) frequency band. This can be achieved, e.g., through insertion of high- and low-pass filters before the photodetectors \cite{Philipp} (see App. \ref{AppBalancedHomodyne} for more details). The (filtered) $\nu$-polarization-photon number operator, $\hat{N}_{\widetilde{\omega},\nu}= \int_{\widetilde{\omega}-\Delta\omega/2}^{\widetilde{\omega}+\Delta\omega/2}\mathrm{d}\omega\, \hat{a}_{\omega,\nu}^{\prime\dag}\hat{a}'_{\omega,\nu}$, can be shown to be proportional to $\hat{\mathfrak{u}}_{\widetilde{\omega}}^{[1]}$ to leading order in the probe amplitude $\alpha$.

\subsection{Ellipsometry Scheme}
We are interested in the quadrature variance of $\bar{\mathfrak{a}}_{\widetilde{\omega}}$ for various ${\widetilde{\omega}}$. As described in the discussion around Eq. \eqref{EqnDetectorVariance}, measurement of the UDW detector's mode, represented by $\hat{\mathfrak{u}}_{\widetilde{\omega}}$ in the present formulation, conveys, within the validity of adopted approximations, most of the information one needs about the sampled field. It therefore allows for the characterization of the main statistical features of the quasiprobability distributions describing the state associated with $\bar{\mathfrak{a}}_{\widetilde\omega}$.

Electro-optic sampling makes use of ellipsometry to implement a functionality similar to homodyning, e.g. both ellipsometry and homodyning rely on the linear superposition of two fields. In this section, we shall focus our discussion on ellipsometry, while further analogy to its mathematically equivalent scheme, the (polarization-based) balanced-homodyne detection, is provided in App. \ref{AppBalancedHomodyne}.
We consider an ellipsometry scheme as depicted in Fig. \ref{FigExperimentalSetup}. 

\begin{figure}[t]\includegraphics[width=0.45\textwidth]{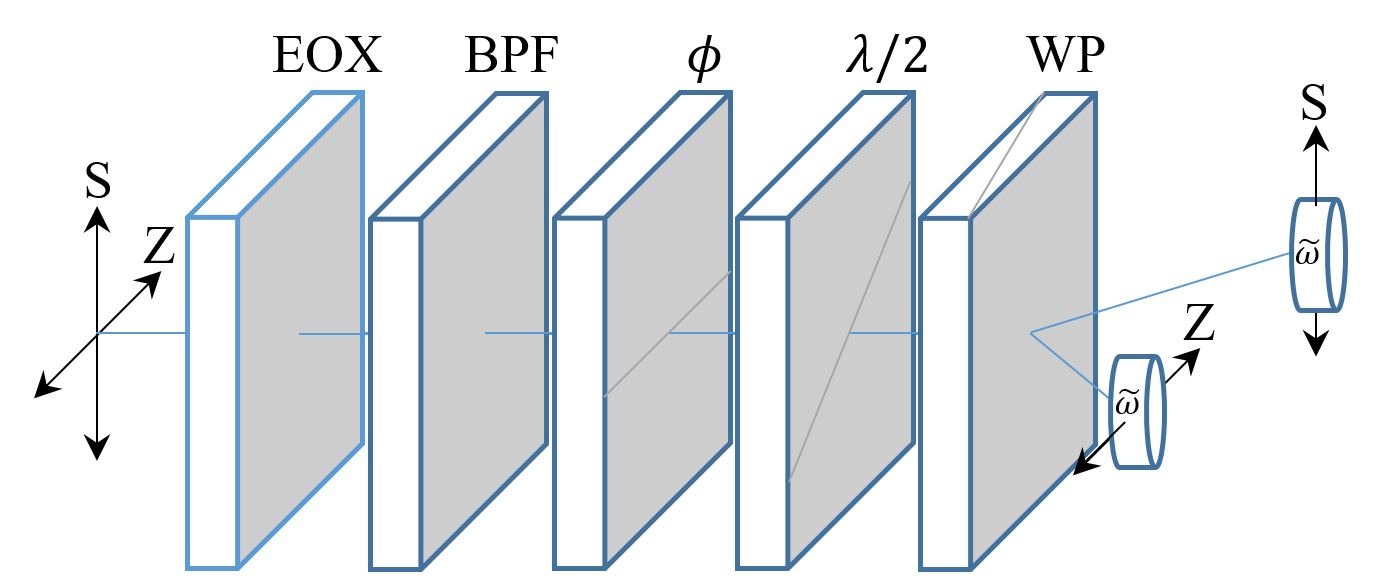}
\caption{The electric field propagates through the {electro-optic crystal} (EOX) for the coherent pulse to induce an interaction with the vacuum. The outgoing field propagates through a narrow band-pass filter (BPF), filtering out outside the narrow-band frequency window $\widetilde\omega-\Delta \omega/2 \leqslant \omega \leqslant \widetilde\omega+\Delta \omega/2$. This is followed by a $\phi_z$ wave-plate, which applies to the field a $\phi$-phase shift in the {$z$-polarization}. Then the field passes through a half-wave plate ($\lambda/2$) at an angle of $\pi/8$ from the horizontal (see diagonal light-grey line) and a Wollaston prism (WP), which physically splits the s and z components of the field. Each output is detected by a photon counter.
}\label{FigExperimentalSetup}
\end{figure}
{This scheme utilizes a $\phi$-waveplate (i.e. a $\phi$-phase shifter) to change the phase of the $z$-polarized field. This is followed by a $\pi$-waveplate at an angle of $\pi/8$ rotated relative to the $z$-axis on the plane perpendicular to the propagation axis (cf. Ref.~\cite{Philipp}). A Wollaston prism then spatially splits the electric field into its $s$ and $z$ components, which are separately measured by photo-detectors. For the measurement of the modes described in this section, the electric field additionally passes through a band-pass filter of width $\Delta \omega$ around $\widetilde\omega$ after the non-linear crystal (EOX). 
The result of such a measurement can be related to Eq. \eqref{EqnDetectorVariance}: $
\braket{(\hat{\mathfrak{u}}^{[1]}_{\widetilde\omega}e^{-i\phi}+\hat{\mathfrak{u}}^{[1]}_{\widetilde\omega}{}^\dagger e^{i\phi})^2}=1+2\sin^2(\theta^{[1]}_{\widetilde\omega})\braket{\bar{\mathfrak{a}}_{\widetilde\omega}^{\dag}\bar{\mathfrak{a}}_{\widetilde\omega}}+2\sin^2(\theta^{[1]}_{\widetilde\omega}) \Re[\braket{\bar{\mathfrak{a}}_{\widetilde\omega}^2\
e^{- 2 i \phi}}]$. In fact, the quadrature operator $\hat{X}^{[1]}_{\widetilde\omega}(\phi)=\hat{\mathfrak{u}}^{[1]}_{\widetilde\omega}e^{-i\phi}+\hat{\mathfrak{u}}^{[1]}_{\widetilde\omega}{}^\dagger e^{i\phi}$ can have its expectation values directly extracted from the electro-optic measurement through \cite{PhysRevX.9.011007}
\begin{equation}
\hat{X}^{[1]}_{\widetilde\omega}(\phi) \approx \frac{\hat{N}_{\widetilde{\omega},z}(\phi)-\hat{N}_{\widetilde{\omega},s}(\phi)}{\sqrt{\braket{\hat{N}_{\widetilde{\omega},z}(\phi)+\hat{N}_{\widetilde{\omega},s}(\phi)}}} \, ,
\label{EOquadrature}
\end{equation}
where $\hat{N}_{\widetilde{\omega},\nu}(\phi)$ can be obtained from $\hat{N}_{\widetilde{\omega},\nu}$ by application of the proper $\phi$-dependent rotation matrices on its annihilation and creation operators.

For the specific cases of $\phi=0$ and $\phi=\pi/2$ with $\bar{\mathfrak{a}}_{\widetilde{\omega}}=\hat{\mathfrak{a}}_{\widetilde{\omega}}$ (i.e. the UDW-detector regime), a slight reformulation of Eq. \eqref{EqnFirstOrderUnitaryEval1}  allows us to write the variances of $\hat{X}^{[1]}_{\widetilde\omega}(\phi)$ in a way similar to Eq. \eqref{EqnUnruhDeWittXu^2}:
\begin{gather}
\begin{aligned}
\Big\langle\! \left(\hat{X}^{[1]}_{\widetilde\omega}(0)\right)^{\!2}\!\Big\rangle &= \cos^2\!\theta^{(1)}_{\widetilde{\omega}}\langle \hat{Q}_{\mathfrak{u}}^2\rangle + \sin^2\!\theta^{(1)}_{\widetilde{\omega}}\langle \hat{Q}_{\mathfrak{a}}^2\rangle\\
&\hspace{-12mm}=1+2\sin^2\!\theta^{(1)}_{\widetilde{\omega}}\braket{\hat{\mathfrak{a}}_{\widetilde\omega}^{\dag}\hat{\mathfrak{a}}_{\widetilde\omega}}
+2\sin^2\!\theta^{(1)}_{\widetilde{\omega}} \Re[\braket{\hat{\mathfrak{a}}_{\widetilde\omega}^2}]\, , \label{min_var}
\end{aligned}
\\
\begin{aligned}
 \Big\langle\!\left(\hat{X}^{[1]}_{\widetilde\omega}\left(\frac{\pi}{2}\right)\right)^{\!2}\!\Big\rangle &= \cos^2\!\theta^{(1)}_{\widetilde{\omega}}\langle \hat{P}_{\mathfrak{u}}^2\rangle + \sin^2\!\theta^{(1)}_{\widetilde{\omega}}\langle \hat{P}_{\mathfrak{a}}^2\rangle
\\&\hspace{-14mm}=1+2\sin^2\!\theta^{(1)}_{\widetilde{\omega}}\braket{\hat{\mathfrak{a}}_{\widetilde\omega}^{\dag}\hat{\mathfrak{a}}_{\widetilde\omega}}
-2\sin^2\!\theta^{(1)}_{\widetilde{\omega}} \Re[\braket{\hat{\mathfrak{a}}_{\widetilde\omega}^2}]\, .  \label{max_var}
\end{aligned}
\end{gather}
$Q$ and $P$ stand for the two orthogonal phase-space quadratures and the subscripts $\mathfrak{u}$ and $\mathfrak{a}$ represent the $\hat{\mathfrak{u}}_{\widetilde\omega}$ and $\hat{\mathfrak{a}}_{\widetilde\omega}$ operators, respectively. In this regime, we obtain a similar result to the case of UDW detector (refer to Eqs. \eqref{EqnUnruhDeWittXu}-\eqref{EqnDetectorVariance}). The pure $\hat{\mathfrak{u}}_{\widetilde{\omega}}$-terms (i.e. $\braket{\hat{Q}_u^2}$ and $\braket{\hat{P}_u^2}$) are associated with the NIR shot noise, which sums up with similar contributions from the $\hat{\mathfrak{a}}_{\widetilde\omega}$-terms to give 1. The latter terms ($\braket{\hat{Q}_{\mathfrak{a}}^2}$ and $\braket{\hat{P}_{\mathfrak{a}}^2}$) contain the ($\widetilde\omega$-dependent) information about the ultrabroadband MIR mode (i.e. subcycle mode) of interest. We note that $\hat{Q}_{\mathfrak{a}}$ and $\hat{P}_{\mathfrak{a}}$ are related via a phase shift of $\pi/2$ on $\hat{\mathfrak{a}}_{\widetilde\omega}$, not on the positive-frequency modes, $\hat{a}_{\Omega}$, due to reasons discussed in Sec. \ref{SecSubcycleProp}.

\subsection{Numerical results}
In order to visualize the analytical results of this section, we consider realistic values for the variables involved in the electro-optical analogue of the UDW detector \citep{Alfred_vacuum,Andrey}. For the probe waveform, Eq. \eqref{waveformE}, we assume a central frequency of $\omega_p/(2\pi)=255 \,\mathrm{THz}$ and a temporal intensity profile with a full-width half maximum of $5.8\, \mathrm{fs}$ (therefore $\sigma_p\approx 203 \,\mathrm{THz}$). The effective cross-sectional area of the beam (waist) is $A= \pi r^2$ with $r=3 \, \mathrm{\mu m}$ and the probe-pulse photon content is $|\alpha|^2=5 \times 10^9$. The length of the zincblende-type crystal is $L=7 \, \mathrm{\mu m}$ and its electro-optic coefficient is taken as $r_{41}=4~\mathrm{pm/V}$ (for the particular case of ZnTe). The refractive index, $n_{\Omega}$, varies only slightly (from 2.55 to 2.59) in the MIR \cite{Andrey}. We utilize a fit for the refractive index in the NIR frequency range $n_{{\omega}}$ \cite{Marple1964} (more details on the refractive index are given in App. \ref{AppNonLinHam}). We set $\Delta \omega/(2\pi)=1 \, \mathrm{THz}$ to ensure that this frequency band is small enough so that the (quasi-)monochromatic approximation is valid, but large enough so that the photon count is {sufficiently} high. Figure \ref{FigPlot1stOrd} shows the maximal and minimal quadrature variances (i.e. $\hat{P}^{[1]}_{\widetilde\omega}$- and $\hat{Q}^{[1]}_{\widetilde\omega}$-variances respectively) 
for various detected central frequencies, $\widetilde{\omega}$.

\begin{figure}[t]
\includegraphics[width=0.45\textwidth]{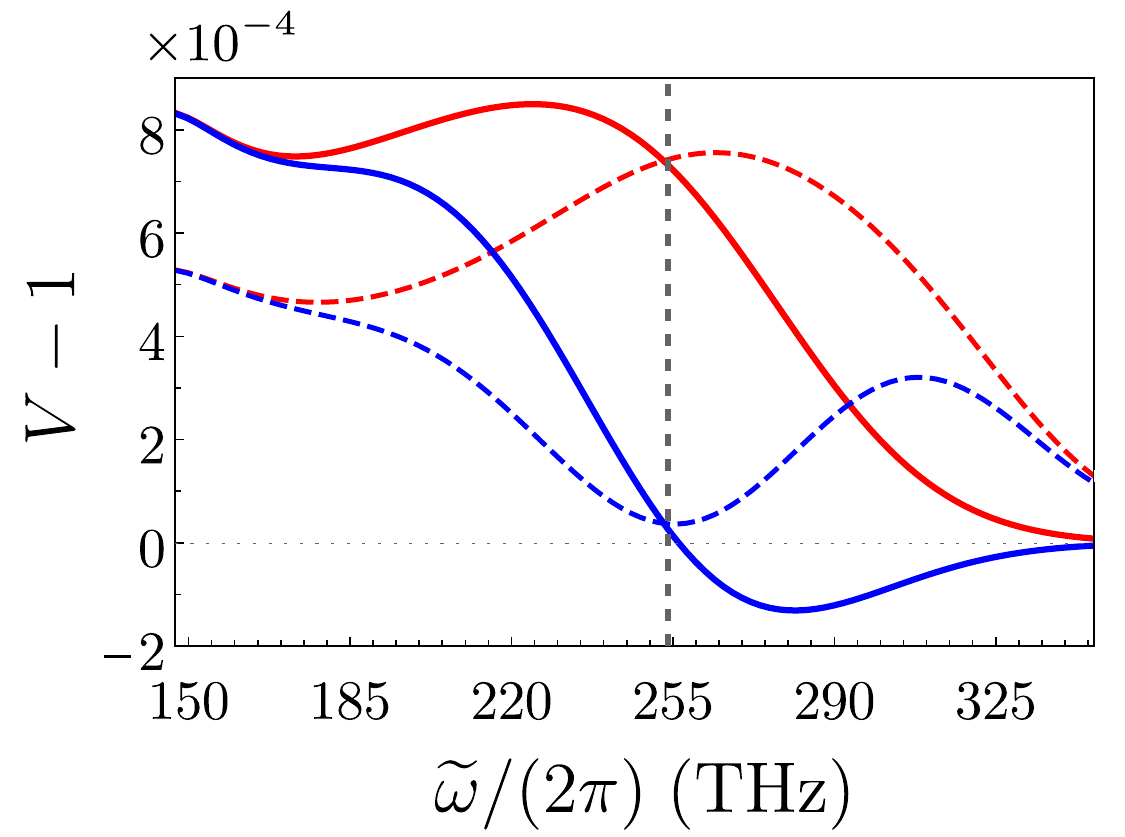}
\caption{Numerical plot of the {$Q$- and $P$-}quadrature variances {obtained} with the first-order unitary evolution method. The top line (solid red) is the {$P$-quadrature} variance, while the bottom line (solid blue) represents the {$Q$-quadrature} one. The dashed lines represent the numerical result when we utilize the standard first-order perturbation theory. In this plot, there is a vertical line near $\omega_p$. To the right of it, the electro-optic sampling can be modelled as {a UDW} detector. 
}
\label{FigPlot1stOrd}
\end{figure}

Through the first-order unitary evolution we are able to delineate two important regimes, portrayed in Fig. \ref{FigPlot1stOrd} via a dividing vertical dashed line. The region on the right-hand side of this line corresponds to the regime in which the interaction can be modelled as a beam-splitter operation between $\hat{\mathfrak{u}}_{\widetilde{\omega}}$ and $\hat{\mathfrak{a}}_{\widetilde{\omega}}$. {We identify this regime as the subcycle probing of the vacuum according to the sampling mechanism of a UDW detector.} In this frequency range, deviations of both quadrature variances from the value of 1 are related to the sampling of virtual particles.
On the other hand, when $\widetilde{\omega}\lesssim \omega_p$, one gets $\bar{\mathfrak{a}}_{\widetilde{\omega}}\to\hat{\mathfrak{a}}_{\widetilde{\omega}}^\dagger$ and therefore the interaction between the two modes is modelled as a squeezing-type operation. The detected particles can be attributed to one half of the photon pairs created by this two-mode squeezing, with the trace over the subspace corresponding to the other half therefore explaining the larger than shot-noise (i.e. $>1$) values for both quadratures in this frequency range (thermalization). It is worth noting that thermalization is also seen in the UDW regime in Fig. \ref{FigPlot1stOrd}: the product of the two variances shows that this is not a minimal-uncertainty state. This thermalization effect is a signature of entanglement breakage/redistribution between virtual particles and the UDW detector, hinting at the possibility of harnessing (vacuum) entanglement from these particles.

For the sake of comparison, we include  in Fig. \ref{FigPlot1stOrd} similar results for the standard first-order perturbation theory as in Ref.~\cite{Andrey}. The key difference between these approaches lies in the prediction of $\hat{\mathfrak{u}}_{\widetilde\omega}$-quadrature measurements that are below unity for some $\widetilde\omega$ values when the first-order unitary approximation is employed. This shows that for $\tilde{\omega}\gtrsim \omega_p $ the entanglement between the involved modes is underestimated when the first-order perturbation theory is applied, while thermalization effects are overestimated. These two contributions compensate each other when integration over frequencies is considered, therefore leading to negligible disagreement between the integrated variances presented here and in Refs.~\cite{Andrey, Thiago, Matthias}. The prediction of sub-shot-noise variance through the first-order unitary approximation represents an important step in understanding the properties of the electromagnetic vacuum, considering sub-shot-noise quadrature variance as one of the trademarks for characterization of quantumness.

\section{Conclusion}
In summary, we characterized a massless bosonic field mode with a subcycle Gaussian profile. Its subcycle character renders it broadband enough to encompass both positive and negative frequencies, resulting in a nonzero photon-number expectation value in the mode, even though the background field state is the Minkowski vacuum state. The photons present in the vacuum are off-shell (virtual) particles and therefore cannot exist outside of very short time intervals. We show that a simple harmonic-oscillator UDW detector interacting with the field through a very fast switching on and off of the interaction can couple to the subcycle Gaussian mode and therefore detect its virtual particles. We then translate this behavior to the language of electro-optic sampling by finding the regime in which the nonlinear electro-optic interaction reproduces the UDW-field interaction Hamiltonian. Since electro-optic sampling is an inherently subcycle technique, the fast switching on and off of the coupling in the optical system is guaranteed by the ultrashort pump that drives the interaction. From these results, it is possible to identify the contributions from virtual particles (i.e from the quantum vacuum itself) to the signal variance detected in quantum electro-optic measurements of the electromagnetic vacuum. 

The comparison between the normalized action of the UDW detector and the nonlinear $\chi^{(2)}$ interaction of the electric field within an optically active crystal is enabled by the introduction of a novel first-order unitary evolution approximation. This method models the nonlinear interaction of the electric field as either a beam-splitter interaction or two-mode squeezing between a detected mode and a subcycle mode. When the beam-splitter interaction dominates, excitations of the subcycle mode are mapped directly onto the probe mode in an analogous way as for the UDW detector. On the other hand, when the two-mode squeezing dominates the process, the interaction is similar to the usual squeezing seen in non-subcycle quantum optics.

Our numerical results for a setting similar to what is found in the literature \cite{Andrey, Alfred_vacuum} allow us to estimate the transition point between regimes in which two-mode squeezing and beam-splitting are predominant, respectively, therefore providing a rigorous delineation of the conditions necessary to achieve a proper mapping of a $\chi^{(2)}$ interaction into a UDW detector. Furthermore, we have proposed an ellipsometry scheme to allow for detection of the first and second moments of the electric-field quadratures of such subcycle modes of the vacuum.

The detection of virtual photons from the vacuum field is a feature of (curved-space) quantum field theory effects such as Unruh-Davies \cite{Unruh} and Hawking radiation \cite{Hawkingeffect}. In these effects, the virtual particle and antiparticle of the pair are separated by a horizon, the Rindler and event horizons for Unruh and Hawking effects, respectively. In such situations the particle and antiparticle of the pair are necessarily delocalized, and hence entangled. For the subcycle UDW detector, a horizon can be introduced via the fast switching on and off of the interaction. This fast switching can decouple the virtual antiparticle from the observed virtual particle, leading in principle to the ability to observe vacuum entanglement between different regions of spacetime \cite{Reznik}. In our work, we do not observe strong decoupling from the entangled virtual antiparticle. The inability to strongly decouple the antiparticles can be traced back to the Gaussian-profile switch. The Gaussian profile leads to a ``soft" horizon in which the probe predominantly detects both the particle and the antiparticle in the same mode. The signature of this effect is the squeezing that we observe in both Fig. \ref{FigPlotSubcycle}
and \ref{FigPlot1stOrd}. As a future research direction, it would be interesting to explore in detail how the statistics of the subcycle mode is affected by the profile of the switching function and hence seek pump profiles that are compatible with the observation of vacuum entanglement.
\section{Acknowledgements}
This work is supported by the Australian Research Council (ARC) under the Centre of Excellence for Quantum Computation and Communication Technology (Grant No. CE170100012). T.L.M.G. and A.S.M. gratefully acknowledge the funding by the Baden-W\"{u}rttemberg Stiftung via the Elite Programme for Postdocs. A.S.M. was also supported by the National Research Foundation of Korea (NRF) grant funded by the Korea government (MSIT) (2020R1A2C1008500). T.L.M.G., A.L. and G.B. acknowledge funding by the Deutsche Forschungsgemeinschaft
(DFG) - Project No. 425217212 - SFB 1432. The authors thank D.V. Seletskiy for helpful discussions.

\onecolumngrid

\appendix
\section{Input-output formalism for an Unruh-DeWitt detector}\label{AppInputOutput}
In this section, we consider the input-output relation of the Unruh De-Witt detector{, with $\hat{u}$ and $\hat{u}'$ being the input and output mode operators,} respectively. The output can be calculated via the Heisenberg evolution of the input {operator. Using} the interaction Hamiltonian (\ref{EqnInterHamUDW}), we introduce the evolution operator
\begin{equation}
\hat{U}_{I,\mathcal{T}}=\mathcal{T}\left[\exp\Big[\frac{-i}{\hslash}\int_{-\infty}^{\infty} \hspace{-3.5 mm}\mathrm{d}\tau \, A\lambda(\tau)\hat{Q}(\tau) \hat{\Pi}\big( t(\tau),x(\tau)\big) \Big]\right]	\, , \label{evolution_no_TO}
\end{equation}
where $\mathcal{T}$ is the time-ordering {protocol}. The output under this unitary operator can be calculated as follows:
\begin{equation}
\hat{u}'= \hat{U}_{I,\mathcal{T}}^{\dag}\hat{u}\hat{U}_{I,\mathcal{T}}	{\, .}
\end{equation}
This can be evaluated via a Magnus expansion, which is difficult to compute non-perturbatively. In certain regimes, the time-ordering effect can be neglected \cite{Silberhorn}, allowing for the following approximation:
\begin{equation}
\hat{U}_{I,\mathcal{T}}\approx \hat{U}_{I}= \exp\left[\frac{-i}{\hslash}\int_{-\infty}^{\infty}\hspace{-3.5 mm} \mathrm{d}\tau\, A\lambda(\tau)\hat{Q}(\tau) \hat{\Pi}\big( t(\tau),x(\tau)\big) \right]	{ \, .}
\end{equation}
The Baker-Hausdorff lemma can now be utilized to compute $\hat{u}'${.} By noting that $\hat{\Pi}$ is Hermitian, we write the unitary operator in the following way:
\begin{equation}
\hat{U}_I=\exp\left[\left(-\frac{i}{\hslash}\int_{-\infty}^{\infty}\hspace{-3.5 mm} \mathrm{d}\tau\,A \lambda(\tau) \hat{\Pi}\big( t(\tau),x(\tau)\big)e^{i\omega_u\tau}\right)\hat{u}^{\dag}-\mathrm{h.c.} \right] {\, .}
\end{equation}
{Inserting Eqs. (\ref{EqnConjufateField}) and (\ref{EqnSwitchingFunction}) into the previous expression and considering a worldline of the form $(t,x)=(\tau, 0)$, we obtain}
\begin{gather}
\begin{aligned}
\hat{U}_I&=\exp\left[\left(-\int_{-\infty}^{\infty} \hspace{-3.5 mm}\mathrm{d}\omega\,\text{sign}(\omega) \eta \sqrt{\frac{|\omega|}{8\pi}}\int_{-\infty}^{\infty}\hspace{-3.5 mm} \mathrm{d}t\, e^{-\sigma_u^2(t-t_u)^2-i(\omega-\omega_u) t} \hat{a}_\omega \right)\hat{u}^{\dag}-\mathrm{h.c.} \right]
\\
&=\exp\left[\left(-\int_{-\infty}^{\infty}\hspace{-3.5 mm} \mathrm{d}\omega\, \text{sign}(\omega) \frac{\eta}{2\sigma_u} \sqrt{\frac{|\omega|}{2}}e^{-{\frac{(\omega-\omega_u)^2}{4\sigma_u^2}}-i(\omega-\omega_u) t_u}
\hat{a}_\omega \right)\hat{u}^{\dag}-\mathrm{h.c.} \right] {.}
\end{aligned}
\end{gather}
It is noticed that the term inside the brackets is similar to the spectral decomposition of a Gaussian profile mode{,} Eq. (\ref{EqnDefoffPOmega}). By setting $\sigma_u=\sigma$, $t_u=t_0$ and $\omega_u=\omega_0$, {the unitary evolution can be cast in the form}
\begin{gather}
\hat{U}_I=\exp\left[\left(\frac{\eta}{2}\sqrt{\frac{\omega_0}{{\sigma}}}\left(\frac{\pi}{2}\right)^{1/4}\int_{-\infty}^{\infty}\hspace{-3.5 mm} \mathrm{d}\omega\, f_{g}(\omega) \hat{a}_\omega \right)\hat{u}^{\dag}-\mathrm{h.c.} \right]{,}
\end{gather}
{or alternatively, b}y setting $\theta_{u}=-\frac{\eta}{2}\sqrt{\frac{\omega_0}{\sigma}}$, in the form
\begin{equation}
\hat{U}_I=\exp\left[\theta_{u} (\hat{a}_{g}\hat{u}^{\dag}-\mathrm{h.c.} )\right]{\, .}
\end{equation}
\section{Decomposition of operators}\label{AppQuadraticOperators}
\subsection{Linear and quadratic operator decomposition}
In this section we demonstrate how linear and quadratic operators can be decomposed in terms of an arbitrary basis set. We introduce a complete discrete orthonormal bosonic basis set $\{\hat{a}_i,\hat{a}_j,...\}$ that satisfies the commutation relations $[\hat{a}_i,\hat{a}_{j}]=0$ and $[\hat{a}_i,\hat{a}_{j}^{\dag}]=\delta^i_{j}${.} 

We first consider an arbitrary operator $\hat{A}$, {that is} {linear in both $\hat{a}_i$ and $\hat{a}^\dagger_i$}. Given that this operator belongs to the space spanned by the complete set $\{\hat{a}_i,\hat{a}_j,...\}$, this operator can be decomposed {as}
\begin{equation}
\hat{A}=\sum_i A_i \hat{a}_i+A_{i}' \hat{a}_i^{\dag}{\, .}
\end{equation}
The prefactors in front of $\hat{a}_i$ can be found utilizing the properties of the commutation relations{,}
\begin{equation}
\begin{aligned}
A_i=[\hat{A},\hat{a}_i^{\dag}],\; A_{i}'=[\hat{a}_i,\hat{A}]{,}
\end{aligned}
\end{equation}
giving the following result:
\begin{equation}
\hat{A}=\sum_i [\hat{A},\hat{a}_i^{\dag}] \hat{a}_i+[\hat{a}_i,\hat{A}] \hat{a}_i^{\dag}{\, .}
\end{equation}
Let us {now consider another operator, $\hat{B}$, quadratic in $\hat{a}_i$ and $\hat{a}^\dagger_i$}. This operator can be decomposed in terms of $\{\hat{a}_i,\hat{a}_j,...\}$ in the following way:
{\begin{equation}
\begin{aligned}
\hat{B}& = \sum_{i} \left[B_{ii} \hat{a}_i\hat{a}_i+B'_{ii}\hat{a}_i^{\dag}\hat{a}_i+B''_{ii}\hat{a}_i^{\dag}\hat{a}_i^{\dag}+\sum_{j>i}\left( B_{ij}\hat{a}_i\hat{a}_j + B'_{ij}\hat{a}_i^{\dag}\hat{a}_j+B'_{ji} \hat{a}_j^{\dag}\hat{a}_i+B''_{ij}\hat{a}_i^{\dag}\hat{a}_j^{\dag}\right)\right] {}\, .
\end{aligned}
\end{equation}
}
Utilizing properties of the commutation relations, we {find}
\begin{equation}
\begin{aligned}
{B_{ij}}& =\frac{1}{1+\delta_j^i} \left[[{\hat{B}},\hat{a}_i^{\dag}],\hat{a}_j^{\dag}\right] {\, ,}
\\
{B''_{ij}}&=\frac{1}{1+\delta_i^j} \left[\hat{a}_j,[\hat{a}_i,{\hat{B}}]\right]{\, ,}
\\
{B'_{ij}}&=\left[[\hat{a}_i,{\hat{B}}],\hat{a}_j^{\dag}\right]{\, .}
\end{aligned}
\end{equation}
\subsection{Parallelization and orthogonalization of a quadratic operator}
In this section we introduce a process we refer to as parallelization/orthogonalization with respect to a quadratic operator. We consider {as reference the operator $\hat{a}_k$, a chosen operator within orthonormal the set $\{\hat{a}_i,\hat{a}_j,...\}$, for which $[\hat{a}_i,\hat{a}^\dagger_j]=\delta^i_j$. We can split a normally ordered quadratic operator $\hat{B}$ into components parallel to $\hat{a}_k$ and components orthogonal to $\hat{a}_k$, namely}
\begin{equation}
\hat{B} = \hat{B}_{\parallel k}+\hat{B}_{\perp k} {\, .}\label{EqnparallelizationandorthogonalBn}
\end{equation}
The {components parallel to $\hat{a}_k$ (i.e. terms that do not commute with either $\hat{a}_k$ or $\hat{a}_k^{\dag}$) can be expressed as:
\begin{equation}
\hat{B}_{\parallel k} =  \sum_{j}\left[ B_{kj}\hat{a}_k\hat{a}_j + B'_{kj}\hat{a}_k^{\dag}\hat{a}_j+(1-\delta^j_{k})B'_{jk} \hat{a}_j^{\dag}\hat{a}_k+B''_{kj}\hat{a}_k^{\dag}\hat{a}_j^{\dag}\right]\, ,
\end{equation}
where $(1-\delta^j_{k})$ avoids double counting when $j=k$. We will refer to this decomposition as parallelization. It can be reduced to a compact form:
\begin{equation}
\hat{B}_{\parallel k}=[\hat{B},\hat{a}^{\dag}_k]\hat{a}_k+\hat{a}_k^{\dag}[\hat{a}_k,\hat{B}]-(B_{kk} \hat{a}_k\hat{a}_k+B'_{kk}\hat{a}_k^{\dag}\hat{a}_k+B''_{kk}\hat{a}_k^{\dag}\hat{a}_k^{\dag})\label{EqnparallelizationBn}\, .
\end{equation}
The component orthogonal to $\hat{a}_k$ (i.e. all terms that commute with both $\hat{a}_k$ and $\hat{a}_k^{\dag}$) has the form:
\begin{equation}
\begin{aligned}
\hat{B}_{\perp k}& = \sum_{i\neq k} \sum_{j\geq i,j\neq k} B_{ij}\hat{a}_i\hat{a}_j + B'_{ij}\hat{a}_i^{\dag}\hat{a}_j+(1-\delta^i_j)B'_{ji} \hat{a}_j^{\dag}\hat{a}_i+B''_{ij}\hat{a}_i^{\dag}\hat{a}_j^{\dag} \, .
\end{aligned}
\end{equation}

This operator can {also} be written in a compact way by rewriting Eq. (\ref{EqnparallelizationandorthogonalBn}) as $\hat{B}_{\perp {k}}=\hat{B}-\hat{B}_{\parallel {k}}$ and substiting Eq. (\ref{EqnparallelizationBn}):
\begin{equation}
\hat{B}_{\perp k}=\hat{B}-\left([\hat{B},\hat{a}^{\dag}_k]\hat{a}_k+\hat{a}_k^{\dag}[\hat{a}_k,\hat{B}]-(B_{kk} \hat{a}_k\hat{a}_k+B_{\breve{k}k}\hat{a}_k^{\dag}\hat{a}_k+B_{\breve{k}\breve{k}}\hat{a}_k^{\dag}\hat{a}_k^{\dag})\right){\, .}
\end{equation}
The second-order parallelization with respect to {both} $\hat{a}_k$ and $\hat{a}_{k'}$ gives
\begin{equation}
\hat{B}_{\parallel k k'} =B_{k k'}\hat{a}_k\hat{a}_{k'} + B'_{k k'}\hat{a}_k^{\dag}\hat{a}_{k'}+(1-\delta^k_{k'})B'_{k'k} \hat{a}_{k'}^{\dag}\hat{a}_k+B''_{kk'}\hat{a}_k^{\dag}\hat{a}_{k'}^{\dag} \label{EqnQuadraticParallelization} \, .
\end{equation}

 Any normal ordered operator {quadratic with respect to the set $\{\hat{a}_i,\hat{a}_j,...\}$ can be therefore} decomposed in the following way:
\begin{equation}
\hat{B}=\sum_{i}\sum_{i'\geqslant i}\hat{B}_{\parallel i i'} {\, .}
\end{equation}
\section{{$n$th}-order unitary evolution}\label{AppNthOrderUni}
The $n$th-order unitary evolution is a simplification of the {non-time-ordered evolution operator}, {Eq. \eqref{evolution_no_TO},} based {on} an approximation to the action{:} $\hat{S} \approx \hat{S}^{[n]}$. This section introduces this formalism in the following manner: in section {\ref{AppnthSecDef}}, we introduce an explicit definition of {the formalism, deriving the constraints on the $n$th order action, $\hat{S}^{[n]}$; in section \ref{AppnthSecQuadAct}, we introduce a simple formula for $\hat{S}^{[n]}$ when $\hat{S}$ is a quadratic operator; in section \ref{AppnthSecConstProof}, we prove that the $\hat{S}^{[n]}$ introduced in sec. \ref{AppnthSecQuadAct} satisfies the constraints imposed on a $n$th-order unitary action, as discussed in sec. {\ref{AppnthSecDef}}.}
\subsection{Definition of the formalism}\label{AppnthSecDef}
In this {section}, we introduce the $n$th order unitary evolution method. We begin by considering an arbitrary unitary operator $\hat{U}_{S}$ of the {form}
\begin{equation}
\hat{U}_{S}= \exp[\hat{S}]{\, ,}
\end{equation}
where $\hat{S}$ is an arbitrary {(normalized by $i\hslash$)} action. The Heisenberg evolution of an arbitrary operator, ${\hat{C}'= \hat{U}_S^{\dag} \hat{C} \hat{U}_S}$
, {obeys the Baker-Hausdorff lemma,
\begin{gather}
\begin{aligned}\label{EqnBakerHausdorffai}
\hat{C}_{}'&=\hat{C}_{}+[\hat{C}_{}, \hat{S}]+\frac{1}{2!}\left[[\hat{C}_{}, \hat{S}],{\hat{S}}\right]+\frac{1}{3!}\left[[[\hat{C}_{},{ \hat{S}}],{ \hat{S}}],{ \hat{S}}\right]+...
\\&= \sum_{n=0}^{\infty}\frac{1}{n!}[\hat{C},\hat{S}]^{(n)} \, .
\end{aligned}
\end{gather}
with $[\hat{C},\hat{S}]^{(n)}=[[\hat{C},\hat{S}]^{(n-1)},\hat{S}]$ and $[\hat{C},\hat{S}]^{(0)}=\hat{C}$.
For nontrivial interaction Hamiltonians (contained in the action), the Heisenberg evolution of $\hat{C}$ is usually} difficult to compute. The $n$th-order unitary evolution simplifies this calculation by approximating the normalized action $\hat{S}$ with the normalized $n$th-order unitary action $\hat{S}^{[n]}$. The nth order unitary evolution of {$\hat{C}$ is defined as
\begin{equation}
\begin{aligned}
\hat{C}^{[n]}= \hat{U}_S^{[n]}{}^{\dag}\hat{C}\hat{U}_S^{[n]}= \sum_{n=0}^{\infty}\frac{1}{n!}[\hat{C},\hat{S}^{[n]}]^{(n)} \, .
\end{aligned}
\end{equation}
The $n$th order normalized action $\hat{S}^{[n]}$ is defined so that the evolution is accurate to at least $n$th order when expanded according to the Baker-Hausdorff lemma}:
\begin{equation}\label{EqnConditionforSn}
{[\hat{C},\hat{S}]^{(m)}
=[\hat{C},\hat{S}^{[n]}]^{(m)}, \; \forall\, m\leqslant n\, .}
\end{equation}
\subsection{$\hat{S}^{[n]}$ for quadratic actions}\label{AppnthSecQuadAct}
In this section, we demonstrate a method to determine a simple $\hat{S}^{[n]}$ for an arbitrary normalized quadratic action:
\begin{equation}
\hat{S}=\int_{-\infty}^{\infty} \hspace{-3.5 mm}\mathrm{d}\omega \mathrm{d}\omega' \, S(\omega,\omega') \hat{a}_{\omega}^{\dag} \hat{u}_{\omega'}-h.c. {\, .}
\end{equation}
{To} keep things general{, we} do not specify the commutation relation between $\hat{a}_{\omega}$ and $\hat{u}_{\omega}$. {$\{\hat{a}_{\omega}, \forall\; \omega \in \mathbb{R}\} \cup \{\hat{u}_{\omega},\forall\; \omega \in \mathbb{R}\}$ is} the set {spanning all (linear) operators one can generate with both $\hat{a}_{\omega}$ and $\hat{u}_{\omega}$}. {S}ince we have not specified the commutation relation{s} between $\hat{a}_{\omega}$ and $\hat{u}_{\omega}$, a general treatment would {allow noncommuting} elements {between $\hat{a}_{\omega}$ and $\hat{u}_{\omega'}$.} {When there are noncommuting terms, the union of the sets would not be an orthonormal set} (this would mean that the union of the two subsets may be overcomplete). {For this reason, we introduce a complete orthornormal discrete set of operators $\{\hat{c}_i,\hat{c}_j,...\}$, so that all $\hat{a}_{\omega}$ and $\hat{u}_{\omega}$ can be written as linear combinations of its elements}.

{We represent the $n$th-order evolution component of the ({$k$}th) element of the set $\{\hat{c}_i,...\}$ (i.e. $\hat{c}_k$), as $\hat{\mathfrak{c}}^{(n)}_k$.} This operator is derived in the following way:
\begin{gather}
\bar{\mathfrak{c}}^{(n)}_{k} = [\bar{\mathfrak{c}}^{(n-1)}_{k},\hat{S}] /(\theta^{(n)}_{{k}}) \, \label{Eqnbarmathfrakcn}
\\
{\theta}_{{k}}^{(n)} = \sqrt{\Big|\big[[\bar{\mathfrak{c}}^{(n-1)}_k,\hat{S}],[\hat{S},\bar{\mathfrak{c}}^{(n-1)}_k{}^{\dag}]\big]\Big|} \, ,
\\
\bar{\mathfrak{c}}^{(n)}_{k} = \begin{cases}
\hat{\mathfrak{c}}^{(n)}_{k}& \text{if } \;[\bar{\mathfrak{c}}^{(n)}_{k},\bar{\mathfrak{c}}_{k}^{(n)}{}^{\dag}] > 0\\
    \hat{\mathfrak{c}}_{k}^{(n)}{}^{\dag}&\text{if}\; [\bar{\mathfrak{c}}^{(n)}_{k},\bar{\mathfrak{c}}_{k}^{(n)}{}^{\dag}] < 0
\\
    0 & \text{otherwise} \, . \label{Eqnbarmathfrakcn2}
\end{cases}
\end{gather}
$\hat{\mathfrak{c}}^{(0)}_{k}=\hat{c}_{k}$. Note that the set of evolution components generated from a chosen $\hat{c}_{k}$, $\{\hat{c}_{k},\bar{\mathfrak{c}}^{(1)}_{k},...,\bar{\mathfrak{c}}^{(n)}_{k} \}$, is not orthonormal. We can orthogonalize this set of operators in the following manner (we will omit the index $k$ hereafter):
\begin{gather}
\bar{c}^{(n)} = \bar{\mathfrak{c}}^{(n)}-\left(\sum_{m<n}[\bar{\mathfrak{c}}^{(n)},\bar{\mathfrak{c}}^{(m)}{^{\dag}}]\bar{\mathfrak{c}}^{(m)}+[\bar{\mathfrak{c}}^{(m)},\bar{\mathfrak{c}}^{(n)}]\bar{\mathfrak{c}}^{(m)}{}^{\dag}\right) \, , \label{orthogonalize_c}
\\
\tilde{c}^{(n)}=\begin{cases}
\hat{c}^{(n)}/ \sqrt{\left|\big[\bar{c}^{(n)},\bar{c}^{(n)}{}^{\dag}\big]\right|} & \text{if } \big[\bar{c}^{(n)},\bar{c}^{(n)}{}^{\dag}\big] > 0\\
\hat{c}^{(n)}{}^{\dag}/ \sqrt{\left|\big[\bar{c}^{(n)},\bar{c}^{(n)}{}^{\dag}\big]\right|} & \text{if } \big[\bar{c}^{(n)},\bar{c}^{(n)}{}^{\dag}\big] < 0\\
0              & \text{otherwise} \, .
\end{cases}
\end{gather}
In other words, $\tilde{c}^{(n)}=\bar{c}^{(n)}/\sqrt{\left|\big[\bar{c}^{(n)},\bar{c}^{(n)}{}^{\dag}\big]\right|}$ for $\big[\bar{c}^{(n)},\bar{c}^{(n)}{}^{\dag}\big]\neq 0$, with $\bar{c}^{(n)}$ being either $\hat{c}^{(n)}$ or $\hat{c}^{(n)}{}^\dagger$ depending on the sign of the commutator.
Utilizing the Schmidt decomposition, we can arbitrarily set (as long as the first {$n$ operators $\hat{c}^{(n)}$ are orthogonal to each other) the first {$n$} operators of a complete orthonormal discrete bosonic set. We therefore introduce the complete orthonormal set{, $\{\hat{c}^{(0)},\hat{c}^{(1)},...,\hat{c}^{(n)}\}\cup \{\hat{d}_{1},\hat{d}_{2},...\}$}. We have set $\{\hat{d}_1, \hat{d}_2,...\}$ to be orthogonal to the set $\{\hat{c}^{(0)},\hat{c}^{(1)},...,\hat{c}^{(n)}\}$}. We then define $\hat{S}^{[n]}$ in the following way:
\begin{equation}\label{EqnS^[n]}
\hat{S}^{[n]} = \sum_{m=0}^{n}\sum_{m'=m}^{n}\hat{S}_{\parallel m m'} (\hat{c}^{(0)},\hat{c}^{(0)}{}^\dagger,\hat{c}^{(1)},\hat{c}^{(1)}{}^\dagger,\ldots,\hat{c}^{(n)},\hat{c}^{(n)}{}^\dagger, ) {\, ,}
\end{equation}
where $\hat{S}_{\parallel m m'}$ is defined in Eq. (\ref{EqnQuadraticParallelization}).
\subsection{Proving the validity of $\hat{S}^{[n]}$}\label{AppnthSecConstProof}
In this section, we prove {that} Eq. \ref{EqnS^[n]} leads to an action {accurate} to at least $n$th order in the Baker-Hausdorff{-}lemma expansion. The full expansion of $\hat{S}$ in terms of the {operators in the basis set $\{\hat{c}^{(0)},\hat{c}^{(1)},...,\hat{c}^{(n)}\}\cup \{\hat{d}_1,\hat{d}_2,...\}$ is
\begin{gather}
 \label{EqnSDecomptoS^[n]+}
\hat{S}=\hat{S}^{[n]}+\hat{S}^{[n]}_\perp+\hat{S}_\perp \, ,
\\
\begin{aligned}
\hat{S}^{[n]}&= \sum_{m=0}^{n}\sum_{m'=m}^{n} \hat{S}_{\parallel mm'}(\hat{c}^{(0)},\hat{c}^{(0)}{}^\dagger,\ldots ) \, ,
\\
\hat{S}_\perp & = \sum_{i=1}^{\infty}\sum_{j=i}^{\infty} \hat{S}_{\parallel i j} (\hat{d}_1, \hat{d}^\dagger_1, \dots ) \, ,
\\
\hat{S}^{[n]}_\perp & = \sum_{m=0}^{n}\sum_{i=1}^{\infty} \hat{S}_{\parallel m i} (\hat{c}^{(0)}, \hat{c}^{(0)}{}^\dagger , \ldots ; \hat{d}_1, \hat{d}^\dagger_1 , \ldots )  \, ,
\end{aligned}
\end{gather}
where the sums with respect to $m, m'$ are over the elements of the set $\{\hat{c}^{(0)},\hat{c}^{(1)},...,\hat{c}^{(n)}\}$, while the sums with respect to $i,j$ are} over the elements of the set $\{\hat{d}_1,\hat{d}_2,...\}$.
\begin{proposition}\label{PropTrueifCond12true}
 {$\hat{S}^{[n]}$ is given by Eq. \ref{EqnS^[n]} if the following conditions are} satisfied:
\begin{subequations}
\begin{gather}
{[\hat{S}_{\perp},\bar{\mathfrak{c}}^{(m)}{}^{\dag}]=[\bar{\mathfrak{c}}^{(m)},\hat{S}_{\perp}]=0, \forall\, {m< n}\label{EqnCondForS^[n]ii}} \, ,
\\
{[\hat{S}^{[n]}_\perp,\bar{\mathfrak{c}}^{(m)}{}^{\dag}]=[\bar{\mathfrak{c}}^{(m)},\hat{S}^{[n]}_\perp]=0, \forall\, {m< n}\label{EqnCondForS^[n]if} \, .}
\end{gather}
\end{subequations}
\end{proposition}
\begin{proof}
Eq. (\ref{EqnS^[n]}) {gives} $\hat{S}^{[n]}$ if it satisfies Eq. (\ref{EqnConditionforSn}). Utilizing Eq. (\ref{EqnSDecomptoS^[n]+}):
\begin{equation}
\begin{aligned}
{[\hat{c},\hat{S}]^{(m)}
=\big[[\hat{c},\hat{S}]^{(m-1)},\hat{S}^{[n]}\big]+\big[[\hat{c},\hat{S}]^{(m-1)},\hat{S}_{\perp}\big]+\big[[\hat{c},\hat{S}]^{(m-1)},\hat{S}^{[n]}_{\perp}\big]\label{EqnacS^{m}} \, . }
\end{aligned}
\end{equation}
We note that {$[\hat{c},\hat{S}]^{(m-1)}$ is generated by the operators in $\{\hat{\mathfrak{c}}^{(m-1)},\hat{\mathfrak{c}}^{(m-1)}{}^{\dag}\}$} by definition (refer to Eq. (\ref{Eqnbarmathfrakcn})). The conditions  (\ref{EqnCondForS^[n]ii}) and (\ref{EqnCondForS^[n]if})
therefore imply:
\begin{equation}
\begin{aligned}
{\big[[\hat{c},\hat{S}]^{(m-1)},\hat{S}_{\perp}\big]=\big[[\hat{c},\hat{S}]^{(m-1)},\hat{S}^{[n]}_\perp\big]=0, \; \forall\, m \leqslant n \, . }
\end{aligned}
\end{equation}
By substituting this result into Eq. (\ref{EqnacS^{m}}) we obtain
\begin{gather}
{[\hat{c},\hat{S}]=[\hat{c},\hat{S}^{[n]}]}
\\
{[\hat{c},\hat{S}]^{(m)}
=\big[[\hat{c},\hat{S}]^{(m-1)},\hat{S}^{[n]}\big] \, ,
\; \forall \, m \leqslant n\, . }
\end{gather}
Where the first equation is an explicit form of the second equation when $m=1$. By the domino effect, we obtain
\begin{equation}
{[\hat{c},\hat{S}]^{(m)}=[\hat{c},\hat{S}^{[n]}]^{(m)}, \; \forall\, m \leqslant n\, , }
\end{equation}
which completes the proof.
\end{proof}
\begin{proposition}
{$[\hat{S}_{\perp},\hat{\mathfrak{c}}^{(m)}{}^{\dag}]=[\hat{\mathfrak{c}}^{(m)},\hat{S}_{\perp}]=0, \forall\, {m\leqslant n}$} is true.\label{PropTrueifCond12true1}
\end{proposition}
\begin{proof}
By construction, {each element of the set $\{\hat{c},\hat{c}^{(1)},...,\hat{c}^{(n)}\}$ is orthogonal to any element in the set $\{\hat{d}_1,\hat{d}_2,...\}$. Utilizing this property, one can show that
\begin{equation}
[\hat{S}_{\perp},\hat{c}^{(m)}{}^{\dag}]=[\hat{c}^{(m)},\hat{S}_{\perp}]=0, \forall\, {m\leqslant n} \label{EqnSiiOrthogonal1} \, .
\end{equation}
This implies that
\begin{equation}
[\hat{S}_{\perp},\sum_{m=0}^n A_m \hat{c}^{(m)}+A'_{m}\hat{c}^{(m)}{}^{\dag}]=0 \,.
\end{equation}
Since any element of the set $\{\hat{c},\hat{\mathfrak{c}}^{(1)},...,\hat{\mathfrak{c}}^{(n)}\}$ can be written as a linear combination of the elements in the set $\{\hat{c},\hat{c}^{(1)},...,\hat{c}^{(n)}\}$, by selecting the correct values for $A_m$ and $A'_{m}$, one can assure that
\begin{equation}
[\hat{S}_{\perp},\hat{\mathfrak{c}}^{(m)}{}^{\dag}]=[\hat{\mathfrak{c}}^{(m)},\hat{S}_{\perp}]=0, \forall\, {m\leqslant n} \label{EqnSiiOrthogonal2} \, ,
\end{equation}
thus completing} the proof.
\end{proof}
\begin{proposition}
{$[\hat{S}^{[n]}_\perp,\hat{\mathfrak{c}}^{(m)}{}^{\dag}]=[\hat{\mathfrak{c}}^{(m)},\hat{S}^{[n]}_\perp]=0, \forall\, {m< n}$} is true.\label{PropTrueifCond12true2}
\end{proposition}
\begin{proof}
We prove this conjecture by noticing {that
\begin{equation}
\begin{aligned}
&[\hat{\mathfrak{c}}^{(m-1)},\hat{S}]= Q \hat{\mathfrak{c}}^{(m)}+ R \hat{\mathfrak{c}}^{(m)}{}^{\dag},\;[\hat{S},\hat{\mathfrak{c}}^{(m-1)}{}^{\dag}]= -R \hat{\mathfrak{c}}^{(m)}-Q \hat{\mathfrak{c}}^{(m)}{}^{\dag},  \forall\;{m < n} \, ,
\end{aligned}\label{Eqn[afm-1,S]isin}
\end{equation}
where $\{Q,R\}\in \mathbb{R}$}. For brevity, we only consider the first part of the equation. Utilizing the result from Prop. \ref{PropTrueifCond12true1}, Eq. (\ref{Eqn[afm-1,S]isin}) can be simplified to
\begin{equation}
{[\hat{\mathfrak{c}}^{(m-1)},{\hat{S}^{[n]}}]+[\hat{\mathfrak{c}}^{(m-1)},\hat{S}^{[n]}_\perp]= Q \hat{\mathfrak{c}}^{(m)}+ R \hat{\mathfrak{c}}^{(m)}{}^{\dag}\label{Eqn[afm-1,S]isinProp2Simplify} \, . }
\end{equation}
We notice that {$\hat{S}^{[n]}_\perp$ only includes terms of the form $\hat{c}^{(m)}\hat{d}_i$, $\hat{c}^{(m)}{}^{\dag}\hat{d}_i$, $\hat{c}^{(m)}\hat{d}^\dagger_i$ or $\hat{c}^{(m)}{}^{\dag}\hat{d}_i^{\dag}$.  Similarly, $\hat{S}^{[n]}$ only includes terms of the form $\hat{c}^{(m)}\hat{c}^{(m')}$, $\hat{c}^{(m)}{}^{\dag}\hat{c}^{(m')}$ or $\hat{c}^{(m)}{}^{\dag}\hat{c}^{(m')}{}^{\dag}$. This allows us to conclude that
\begin{equation}
\begin{gathered}
{} [\hat{c}^{(m)},\hat{S}^{[n]}_\perp]=\sum_{i=1}^{\infty} Q_{mi} \hat{d}_{i}+ R_{mi}\hat{d}_i^{\dag} \, ,
\\
{} [\hat{c}^{(m)},\hat{S}^{[n]}]=\sum_{m'=0}^{n} Q'_{m m'} \hat{c}^{(m)}+(1+\delta_m^{m'}) R'_{mm'}\hat{c}^{(m')}{}^{\dag} \, .
\label{Eqn[afm-1,S]Separately}
\end{gathered}
\end{equation}
B}y substituting Eq. (\ref{Eqn[afm-1,S]Separately}) into Eq. (\ref{Eqn[afm-1,S]isinProp2Simplify}) and noticing how the right-hand side does not contain {$\hat{d}_i$} terms, we find the following:
\begin{equation}
{[\hat{S}^{[n]}_\perp,\hat{\mathfrak{c}}^{(m)}{}^{\dag}]=[\hat{\mathfrak{c}}^{(m)},\hat{S}^{[n]}_\perp]=0, \forall\, {m< n}  .}
\end{equation}
This completes the proof.
\end{proof}
Combining Prop. (\ref{PropTrueifCond12true}, \ref{PropTrueifCond12true1}, \ref{PropTrueifCond12true2}), we prove that Eq. (\ref{EqnS^[n]}) is a valid $n$th-order unitary action, satisfying Eq. (\ref{EqnConditionforSn}).
\section{Hamiltonian for a nonlinear field interaction}\label{AppNonLinHam}
In this section we derive Eq. (\ref{EqnSchi2Compact1}) from the following action:
\begin{gather}
\hat{S}_{\chi}= -\frac{i}{\hslash} \int_{-\infty}^{\infty}\hspace{-3.5 mm} \mathrm{d}x\mathrm{d}t \,\lambda \alpha_p(t,x)\left(
\hat{E}_{\Omega}(t,x)+ \hat{E}_{\omega}(t,x)\right)^2\text{rect}\left(\frac{x}{L}\right){\, .}
\end{gather}
By the rotating-wave approximation, we simplify the equation to the following {form}:
\begin{gather}
\hat{S}_{\chi}= -\frac{2i}{\hslash} \int_{-\infty}^{\infty} \hspace{-3.5 mm}\mathrm{d}x\mathrm{d}t \,\lambda \alpha_p(t,x)\hat{E}_{\Omega}(t,x)\hat{E}_{\omega}(t,x)\text{rect}\left(\frac{x}{L}\right) { \, .}
\end{gather}
The Fourier decomposition of these operators leads {to}
\begin{gather}
\hat{S}_{\chi}=  \frac{i \alpha\lambda }{2 \pi A  c \epsilon_0}\int\hspace{-0.5 mm}\mathrm{d}\vec{\omega}\, \text{sign}(\omega \Omega)\sqrt{\frac{|\omega\Omega|}{n_{\omega}n_\Omega}} E_p(\omega_p)
\hat{a}_\omega \hat{a}_{\Omega} \int_{-L/2}^{L/2}\hspace{-3.5 mm}{d}x \, e^{-i(\omega_p n_{\omega_p}+\omega n_\omega+\Omega n_{\Omega})\frac{x}{c}}\int_{-\infty}^{\infty} \hspace{-3.5 mm}\mathrm{d}t \, e^{-i(\omega_p+\omega+\Omega)t}
-h.c. {\, ,}
\end{gather}
where $\mathrm{d}\vec{\omega}=\mathrm{d}\omega_p\mathrm{d}\omega\mathrm{d}\Omega$. The {integral over time can be performed to give}
\begin{gather}
\begin{aligned}
\int_{-\infty}^{\infty}\hspace{-3.5 mm} \mathrm{d}t \, e^{-i(\omega_p+\omega+\Omega)t} =2 \pi \delta(\omega_p+\omega+\Omega) {\, .}
\end{aligned}
\end{gather}
By integrating over $\omega_p$, we set $\omega_p=-\omega-\Omega$:
\begin{gather}
\begin{aligned}
\hat{S}_{\chi}=   \frac{i \alpha\lambda}{ A   c \epsilon_0}\int\hspace{-0.5 mm}\mathrm{d}\omega\mathrm{d}\Omega\, &\text{sign}(\omega \Omega)\sqrt{\frac{|\omega\Omega|}{n_{\omega}n_\Omega}} E_p(-\omega-\Omega)
\hat{a}_\omega \hat{a}_{\Omega}  \int_{-L/2}^{L/2}\hspace{-3.5 mm}{d}x \, e^{-i( \omega (n_\omega-n_{\omega+\Omega})+\Omega (n_{\Omega}-n_{\omega+\Omega}))\frac{x}{c}}-h.c.
\end{aligned}
\end{gather}
Where we have used the property that $n_\omega=n_{-\omega}$. {Integration over the space coordinate gives}
\begin{equation}
\int_{-L/2}^{L/2}\hspace{-3.5 mm}{d}x \, e^{-i( \omega (n_\omega-n_{\omega+\Omega})+\Omega (n_{\Omega}-n_{\omega+\Omega}))\frac{x}{c}}=L\, \text{sinc}\left\{\frac{L}{2c}\big[ \omega (n_\omega-n_{\omega+\Omega})+\Omega (n_{\Omega}-n_{\omega+\Omega})\big]\right\} { \, .}
\end{equation}
This is proportional to the phase-matching {funcion} (i.e. $\zeta_{\omega,\Omega}$). Utilizing this result, we obtain
\begin{gather}
\begin{aligned}
\hat{S}_{\chi}=  \frac{i \alpha\lambda  L}{ A  c \epsilon_0}\int\hspace{-0.5 mm}\mathrm{d}\omega\mathrm{d}\Omega \, \text{sign}(\omega \Omega)\sqrt{\frac{|\omega\Omega|}{n_{\omega}n_\Omega}} E_p(-\omega-\Omega)
\hat{a}_\omega \hat{a}_{\Omega}  \text{sinc}\left\{\frac{L}{2c}\big[ \omega (n_\omega-n_{\omega+\Omega})+\Omega (n_{\Omega}-n_{\omega+\Omega})\big]\right\}-h.c. {\, . }
\end{aligned}
\end{gather}
By setting one of the dummy variable as $\omega \rightarrow -\omega$:
\begin{gather}
\begin{aligned}
\hat{S}_{\chi}=  - \frac{i \alpha\lambda  L}{A   c \epsilon_0}\int\hspace{-0.5 mm}\mathrm{d}\omega\mathrm{d}\Omega \,  \text{sign}(\omega \Omega)\sqrt{\frac{|\omega\Omega|}{n_{\omega}n_\Omega}} E_p(\omega-\Omega)
\hat{a}_\omega ^{\dag}\hat{a}_{\Omega}  
\text{sinc}\left\{\frac{L}{2c}\big[ \Omega (n_{\Omega}-n_{\omega-\Omega})-\omega (n_\omega-n_{\omega-\Omega})\big]\right\}-h.c. {\, .}
\end{aligned}
\end{gather}
{S}etting the dummy variable $\omega \rightarrow -\omega$ and $\Omega\rightarrow -\Omega$ for the Hermitian conjugate, we obtain
\begin{gather}
\begin{aligned}
\hat{S}_{\chi}=  - \frac{i \lambda  L}{A  \gamma}\int\hspace{-0.5 mm}\mathrm{d}\omega\mathrm{d}\Omega \, & \text{sign}(\omega \Omega)\sqrt{\frac{|\omega\Omega|}{n_{\omega}n_\Omega}} \left[\alpha E_p(\Omega-\omega)+\alpha^*E^*_p(\omega-\Omega)\right]
\text{sinc}\left\{\frac{L}{2c}\big[ \Omega (n_{\Omega}-n_{\omega-\Omega})-\omega (n_\omega-n_{\omega-\Omega})\big]\right\}\hat{a}_\omega ^{\dag}\hat{a}_{\Omega}  {\, .}
\end{aligned}
\end{gather}
{With the aid of} Eq. (\ref{EqnPhaseMatching1}) and (\ref{EqnPhaseMatching2})  we write this equation in a compact form:
\begin{gather}
\hat{S}_{\chi}=  \int\hspace{-0.5 mm}\mathrm{d}\Omega\,\mathrm{d}\omega\, \alpha_{p}(\omega-\Omega)\zeta_{\omega,\Omega}\hat{a}_{\Omega} \hat{a}_{\omega}^{\dag} {\, .}
\end{gather}
\section{Numerical results on the second-order unitary evolution} \label{AppNumerics}
\subsection{Second-order unitary evolution}
Let us define the {operator $\hat{\mathfrak{u}}_{\widetilde\omega}=\int \hspace{-0.5 mm}\mathrm{d}\omega \,\frac{1}{\sqrt{\Delta\omega}}\Pi(\frac{\widetilde{\omega}-\omega}{\Delta \omega})\hat{a}_{\omega}$.
We are interested in the second-order unitary evolution under the action
\begin{equation}
\begin{aligned}
\hat{S}_\chi&=\int_{-\infty}^{\infty}\hspace{-3.5 mm} \mathrm{d}\omega \mathrm{d}\Omega \, S_\chi(\Omega,\omega)\hat{a}_{\Omega} \hat{a}_{\omega}^{\dag}  \, .
\end{aligned}
\end{equation}
The evolution of $\hat{\mathfrak{u}}_{\widetilde{\omega}}$ according to the Baker-Hausdorff-lemma expansion to second order reads
\begin{equation}
\hat{\mathfrak{u}}_{\widetilde{\omega}}{}' \approx \hat{\mathfrak{u}}_{\widetilde{\omega}}+[\hat{\mathfrak{u}}_{\widetilde{\omega}},\hat{S}]+\frac{1}{2!}[[ \hat{\mathfrak{u}}_{\widetilde{\omega}},\hat{S}],\hat{S}] \, .
\end{equation}
We can then use the approach defined by Eqs. \eqref{Eqnbarmathfrakcn}-\eqref{Eqnbarmathfrakcn2} to write}
\begin{equation}\label{Eqn2ndOrdPer}
{\hat{\mathfrak{u}}_{\widetilde{\omega}}{}' \approx  \hat{\mathfrak{u}}_{\widetilde{\omega}}{}+\theta^{(1)}_{\widetilde{\omega}}\bar{\mathfrak{u}}^{(1)}_{\widetilde{\omega}}+\frac{1}{2}\theta^{(1)}_{\widetilde{\omega}}\theta^{(2)}_{\widetilde{\omega}}\bar{\mathfrak{u}}^{(2)}_{\widetilde{\omega}} \, . }
\end{equation}
We note that, in the paper $\bar{\mathfrak{u}}_{\widetilde\omega}^{(1)}=\bar{\mathfrak{a}}_{\widetilde\omega}$. {One can orthonormalize these operators with the aid of Eq. \eqref{orthogonalize_c}, with $\bar{u}^{(1)}_{\widetilde\omega}$ and $\bar{u}^{(2)}_{\widetilde{\omega}}$ being the so obtained orthonormalized first- and second-order mode operators.}
{Eq. (\ref{EqnS^[n]}), written in terms of these operators, reads (note that energy conservation does not allow for quadratic terms on the same mode operator, therefore ruling out $\left(\hat{S}_{\chi }\right)_{\parallel ii}$ terms):}
\begin{align}
\hat{S}_{\chi}^{[2]}&=\sum_{m=0}^{2}\sum_{m'=m}^{2} \left( \hat{S}_{\chi}\right)_{\parallel mm'}
\\&=\left(\hat{S}_{\chi }\right)_{\parallel 01}+\left(\hat{S}_{\chi}\right)_{\parallel 12}
\\&=\left(\theta^{(1)}_{\widetilde\omega} {\bar{\mathfrak{u}}}^{(1)}_{\widetilde{\omega}} {\hat{{\mathfrak{u}}}}^{{(0)}}_{\widetilde{\omega}}{}^{\dag}-h.c.\right)+[\bar{\mathfrak{\mathfrak{u}}}^{(1)}_{\widetilde{\omega}},\bar{\mathfrak{\mathfrak{u}}}^{(1)}_{\widetilde{\omega}}{}^{\dag}] \left(\theta^{(2)}_{\widetilde\omega} {\bar{\mathfrak{u}}}^{(1)}_{\widetilde{\omega}}{}^{\dag}\tilde{\mathfrak{u}}^{(2)}_{\widetilde{\omega}}-h.c.\right)
\\
&={[\bar{\mathfrak{u}}^{(1)}_{\widetilde{\omega}},\bar{\mathfrak{u}}^{(1)}_{\widetilde{\omega}}{}^{\dag}]\theta^{(2)}_{\widetilde{\omega}} (\bar{\mathfrak{u}}^{(1)}_{\widetilde{\omega}}{}^{\dag}\bar{\mathfrak{u}}^{(2)}_{\widetilde{\omega}}-\bar{\mathfrak{u}}^{(2)}_{\widetilde{\omega}}{}^{\dag}\bar{\mathfrak{u}}_{\widetilde{\omega}}^{{(1)}})\, . }
\end{align}
{The evolution of $\hat{\mathfrak{u}}_{\widetilde\omega}$ under the second-order unitary evolution is therefore given by
\begin{gather}\label{Eqn2ndOrdUni}
\hat{\mathfrak{u}}_{\widetilde\omega}'=\begin{cases}
\hat{\mathfrak{u}}_{\widetilde\omega}+[\hat{\mathfrak{u}}_{\widetilde\omega},\hat{\mathfrak{u}}^{(2)}_{\widetilde\omega}{}^{\dag}]\left((\cos(\theta^{(2)}_{\widetilde\omega})-1)\hat{\mathfrak{u}}^{(2)}_{\widetilde\omega}{}-\sin(\theta^{(2)}_{\widetilde\omega})\hat{\mathfrak{u}}^{(1)}_{\widetilde\omega}\right), & \text{if }[\bar{\mathfrak{u}}^{(1)}_{\widetilde{\omega}},\bar{\mathfrak{u}}^{(1)}_{\widetilde{\omega}}{}^{\dag}]=1, \;[\bar{\mathfrak{u}}^{(2)}_{\widetilde{\omega}},\bar{\mathfrak{u}}^{(2)}_{\widetilde{\omega}}{}^{\dag}]=1 \, ;
\\
\hat{\mathfrak{u}}_{\widetilde\omega}+[\hat{\mathfrak{u}}^{(2)}_{\widetilde\omega},\hat{\mathfrak{u}}_{\widetilde\omega}]\left((\cosh(\theta^{(2)}_{\widetilde\omega})-1)\hat{\mathfrak{u}}^{(2)}_{\widetilde\omega}{}^{\dag}{}+\sinh(\theta^{(2)}_{\widetilde\omega})\hat{\mathfrak{a}}^{(1)}_{\widetilde\omega}\right), & \text{if }[\bar{\mathfrak{u}}^{(1)}_{\widetilde{\omega}},\bar{\mathfrak{u}}^{(1)}_{\widetilde{\omega}}{}^{\dag}]=1, \;[\bar{\mathfrak{u}}^{(2)}_{\widetilde{\omega}},\bar{\mathfrak{u}}^{(2)}_{\widetilde{\omega}}{}^{\dag}]=-1 \, ;
\\
\hat{\mathfrak{u}}_{\widetilde\omega}+[\hat{\mathfrak{u}}_{\widetilde\omega},\hat{\mathfrak{u}}_{\widetilde\omega}^{(2)}{}^{\dag}]\left((\cosh(\theta^{(2)}_{\widetilde\omega})-1)\hat{\mathfrak{u}}^{(2)}_{\widetilde\omega}{}+\sinh(\theta^{(2)}_{\widetilde\omega})\hat{\mathfrak{u}}^{(1)}_{\widetilde\omega}{}^{\dag}\right), & \text{if }[\bar{\mathfrak{u}}^{(1)}_{\widetilde{\omega}},\bar{\mathfrak{u}}^{(1)}_{\widetilde{\omega}}{}^{\dag}]=-1, \;[\bar{\mathfrak{u}}^{(2)}_{\widetilde{\omega}},\bar{\mathfrak{u}}^{(2)}_{\widetilde{\omega}}{}^{\dag}]=1 \, ;
\\
\hat{\mathfrak{u}}_{\widetilde\omega}+[\hat{\mathfrak{u}}^{(2)}_{\widetilde\omega},\hat{\mathfrak{u}}_{\widetilde\omega}]\left((\cos(\theta^{(2)}_{\widetilde\omega})-1)\hat{\mathfrak{u}}^{(2)}_{\widetilde\omega}{}^{\dag}-\sin(\theta^{(2)}_{\widetilde\omega})\hat{\mathfrak{u}}^{(1)}_{\widetilde\omega}{}^{\dag}\right), & \text{if }[\bar{\mathfrak{u}}^{(1)}_{\widetilde{\omega}},\bar{\mathfrak{u}}^{(1)}_{\widetilde{\omega}}{}^{\dag}]=-1, \;[\bar{\mathfrak{u}}^{(2)}_{\widetilde{\omega}},\bar{\mathfrak{u}}^{(2)}_{\widetilde{\omega}}{}^{\dag}]=-1 \, .
\end{cases}
\end{gather}
By introducing the expressions
\begin{gather}
M^{(0)}_{\widetilde{\omega}}=[\hat{\mathfrak{u}}_{\widetilde\omega},\hat{\mathfrak{u}}_{\widetilde\omega}^{(2)}{}^{\dag}]+[\hat{\mathfrak{u}}^{(2)}_{\widetilde\omega},\hat{\mathfrak{u}}_{\widetilde\omega}] \, ,
\\
M^{({2})}_{\widetilde{\omega}}= \frac{|[\bar{\mathfrak{u}}^{(1)}_{\widetilde{\omega}}+\bar{\mathfrak{u}}^{(2)}_{\widetilde{\omega}},\bar{\mathfrak{u}}^{(1)}_{\widetilde{\omega}}{}^{\dag}+\bar{\mathfrak{u}}^{(2)}_{\widetilde{\omega}}{}^{\dag}]|}{2}\cos(\theta^{(2)}_{\widetilde\omega})+\frac{|[\bar{\mathfrak{u}}^{(1)}_{\widetilde{\omega}}{+}\bar{\mathfrak{u}}^{(2)}_{\widetilde{\omega}},\bar{\mathfrak{u}}^{(1)}_{\widetilde{\omega}}{}^{\dag}-\bar{\mathfrak{u}}^{(2)}_{\widetilde{\omega}}{}^{\dag}]|}{2}\cosh(\theta^{(2)}_{\widetilde\omega}) \, ,
\\
M^{({1})}_{\widetilde{\omega}}= -\frac{|[\bar{\mathfrak{u}}^{(1)}_{\widetilde{\omega}}+\bar{\mathfrak{u}}^{(2)}_{\widetilde{\omega}},\bar{\mathfrak{u}}^{(1)}_{\widetilde{\omega}}{}^{\dag}+\bar{\mathfrak{u}}^{(2)}_{\widetilde{\omega}}{}^{\dag}]|}{2}\sin(\theta^{(2)}_{\widetilde\omega})+\frac{|[\bar{\mathfrak{u}}^{(1)}_{\widetilde{\omega}}{+}\bar{\mathfrak{u}}^{(2)}_{\widetilde{\omega}},\bar{\mathfrak{u}}^{(1)}_{\widetilde{\omega}}{}^{\dag}-\bar{\mathfrak{u}}^{(2)}_{\widetilde{\omega}}{}^{\dag}]|}{2}\sinh(\theta^{(2)}_{\widetilde\omega}) \, ,
\end{gather}
we simplify Eq. (\ref{Eqn2ndOrdUni}) to
\begin{equation}
\hat{\mathfrak{u}}'_{\widetilde{\omega}}=\hat{\mathfrak{u}}_{\widetilde{\omega}}+M^{(0)}_{\widetilde{\omega}}\left((M^{(2)}_{\widetilde{\omega}}-1)\bar{\mathfrak{u}}^{(2)}_{\widetilde{\omega}}{}+M^{({1})}_{\widetilde{\omega}}\bar{\mathfrak{u}}^{(1)}_{\widetilde{\omega}}\right) \, .\label{Eqn2ndOrdUniEvolSimple}
\end{equation}
By utilizing the properties of the commutator, namely

\begin{equation}
\begin{aligned}
[\bar{\mathfrak{u}}^{(2)}_{\widetilde{\omega}},\bar{\mathfrak{u}}^{(2)}_{\widetilde{\omega}}{}^{\dag}]=1 \Rightarrow {}\begin{cases}
[\hat{\mathfrak{u}}_{\widetilde\omega},\bar{\mathfrak{u}}^{(2)}_{\widetilde\omega}{}^{\dag}]=-[\bar{\mathfrak{u}}^{(1)}_{\widetilde{\omega}},\bar{\mathfrak{u}}^{(1)}_{\widetilde{\omega}}{}^{\dag}]\theta^{(1)}_{{\widetilde\omega}}/\theta^{(2)}_{\widetilde\omega} \, ,
\\
[\hat{\mathfrak{u}}^{(2)}_{\widetilde\omega},\hat{\mathfrak{u}}_{\widetilde\omega}]=0 \, ,
\end{cases}
\\
[\bar{\mathfrak{u}}^{(2)}_{\widetilde{\omega}},\hat{\mathfrak{u}}^{(2)}_{\widetilde{\omega}}{}^{\dag}]=-1 \Rightarrow {}
\begin{cases}
[\hat{\mathfrak{u}}_{\widetilde\omega},\hat{\mathfrak{u}}^{(2)}_{\widetilde\omega}{}^{\dag}]=0 \, ,
\\
[\hat{\mathfrak{u}}^{(2)}_{\widetilde\omega},\hat{\mathfrak{u}}_{\widetilde\omega}]=[\bar{\mathfrak{u}}^{(1)}_{\widetilde{\omega}},\bar{\mathfrak{u}}^{(1)}_{\widetilde{\omega}}{}^{\dag}]\theta^{(1)}_{\widetilde\omega}/\theta^{(2)}_{\widetilde\omega} \, ,
\end{cases}
\end{aligned}
\end{equation}
we} show that all four cases of Eq. (\ref{Eqn2ndOrdUni}), coincide with Eq. (\ref{Eqn2ndOrdPer}) up to second order in {$\theta_{\widetilde\omega}^{(2)}$}. This means that $\hat{S}^{[2]}_{\chi}$ is a proper second order normalized action, giving results accurate to at least second order in perturbation theory.
 \begin{figure}[t!]
\begin{center}
\includegraphics[width=0.7\textwidth]{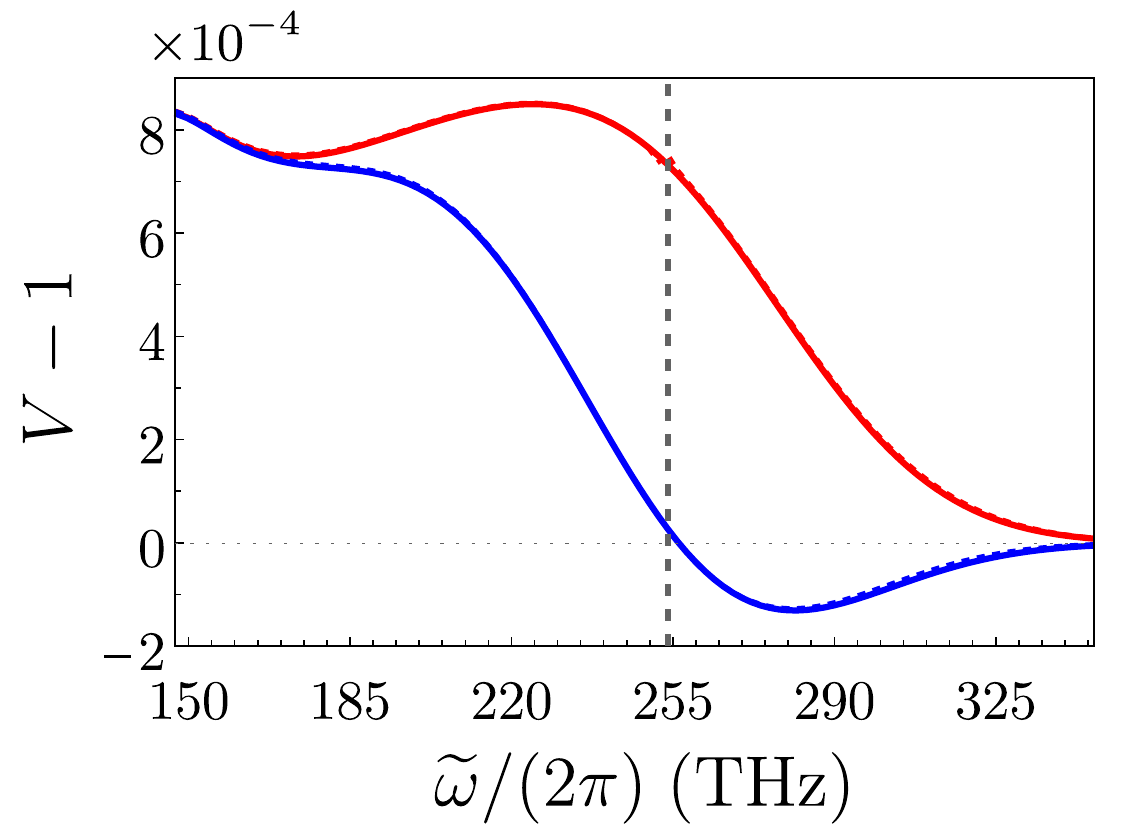}
\caption{Numerical plot of the X and P quadrature variance for both first and second-order unitary evolution theory. The top line (the red line) is the P-quadrature variance, while the bottom line (the blue line) is the X-quadrature variance. The dotted line corresponds to the second order results.}
\label{FigPlot2ndOrd}
\end{center}
\end{figure}

\begin{figure}[t]
\begin{center}
\includegraphics[width=0.7\textwidth]{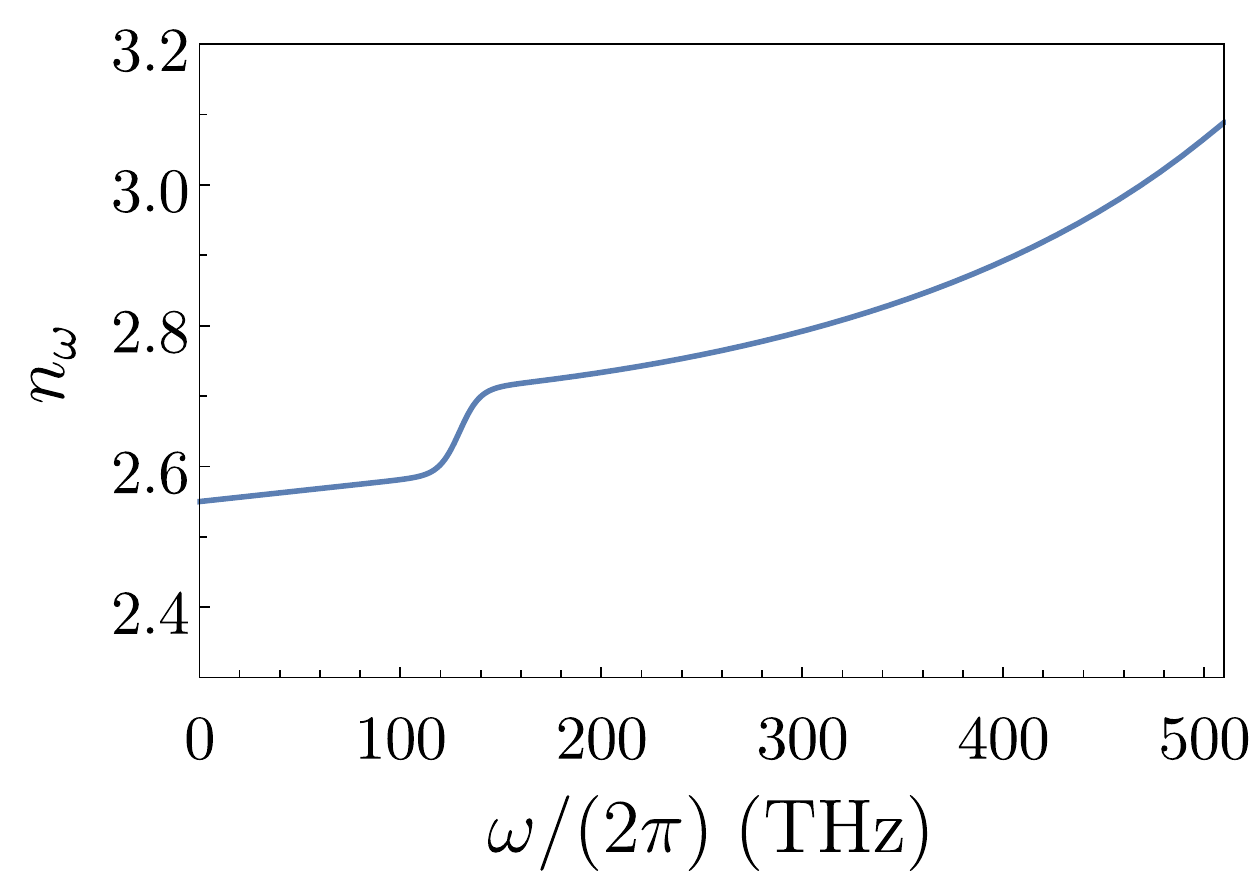}
\caption{The frequency dependent refractive index of the EOX crystal is plotted. We have utilized the (simplified) model of \cite{Andrey} for the MIR regime. For the NIR regime, we have utilized the fit provided by \cite{Marple1964}. We have continuously connected the two models.}
\label{FigRefractiveIndex}
\end{center}
\end{figure}

By utilizing Eq. (\ref{Eqn2ndOrdUniEvolSimple}) we compute the quadrature variance. A numerical plot of the quadrature variance is shown in Fig. \ref{FigPlot2ndOrd}, with the refractive index shown in Fig. ~\ref{FigRefractiveIndex}.
It is found that the difference between the first and second-order unitary evolution method is negligible. There is a minor deviation between the two prediction near $\omega_p$. This is the point where {$\theta_{\widetilde\omega}^{(2)}$} becomes large, and the predictions made by second-order unitary evolution method starts to fail. Higher order calculations are required to have a more accurate model in this regime.
\section{Homodyne detection with electro-optic sampling}\label{AppBalancedHomodyne}
In this section, we consider the mathematics of the ellipsometry scheme. The final state of the electric field{, $\ket{\psi}$, after going through the various crystals can be written in the Schrödinger picture as
$\ket{\psi}=\hat{U}_{\phi_z}\hat{U}_{\chi}\ket{0}$, with
\begin{equation}
\hat{U}_{\phi_z}= \exp\Big(i \phi \int_0^{\infty}\hspace{-3.5 mm}\mathrm{d}\omega \,  \hat{a}_{\omega,z}^{\dag}\hat{a}_{\omega,z}\Big) \, .
\end{equation}
being the evolution operator corresponding to the action of a $\phi$-wave plate. In the Heisenberg picture, on the other hand, evolution is considered on the operators themselves, so that $\hat{a}_{\omega,\nu}$ (with polarization states $\nu=s,z$) is evolved to the form
\begin{equation}
\hat{a}''_{\omega, \nu}=(\hat{U}_{\phi_z}\hat{U}_{\chi})^{\dag}\hat{a}_{\omega,S}\hat{U}_{\phi_z}\hat{U}_{\chi} \, ,
\end{equation}
with
$\hat{U}_{\phi_z}^{\dag}\hat{a}_{\omega,z}\hat{U}_{\phi_z}=e^{i\phi}\hat{a}_{\omega,z}$ and $\hat{U}_{\phi_z}^{\dag}\hat{a}_{\omega,s}\hat{U}_{\phi_z}=\hat{a}_{\omega,s}$. Adopting the notation $\hat{a}_{\omega,S}'= \hat{U}_{\chi}^{\dag}\hat{a}_{\omega,S}\hat{U}_{\chi}$ and considering the alternative pair of annihilation operators $\hat{a}_{\omega,a}=\frac{1}{\sqrt{2}}(\hat{a}_{\omega,s}+\hat{a}_{\omega,z})$ and $\hat{a}_{\omega,b}=\frac{1}{\sqrt{2}}(\hat{a}_{\omega,z}-\hat{a}_{\omega,s})$ for the polarization axes $a$ and $b$ rotated by $\pi/4$ relative to $s$ and $z$, one finds:
\begin{gather}
\hat{a}_{\omega,a}''=\frac{1}{\sqrt{2}}(\hat{a}_{\omega,z}'e^{i\phi}+\hat{a}_{\omega,\nu}') \, ,
\\
\hat{a}_{\omega,b}''=\frac{1}{\sqrt{2}}(\hat{a}_{\omega,z}'e^{i\phi}-\hat{a}_{\omega,\nu}')\, .
\end{gather}

At the output, we consider the measurement of the filtered photon-number operators defined as:
\begin{equation}
\hat{N}_{\widetilde{\omega},\nu}= \int^{\tilde{\omega}+\Delta\omega /2}_{\tilde{\omega}-\Delta\omega /2}\hspace{-3.5 mm}\mathrm{d}\omega \, \hat{a}_{\omega,\nu}^{\dag}\hat{a}_{\omega,\nu} \, .
\end{equation}
The integration domain corresponds to a photodetection limited to the frequency band of width $\Delta \omega$ centered at $\widetilde\omega$. Such a measurement can be achieved with the inclusion of a band-pass filter before detection. The Wollaston prism isolates the particle in the $a$ and $b$ polarisation, allowing the detection of $\hat{N}_{\widetilde{\omega},a}''$ and $\hat{N}_{\widetilde{\omega},b}''$, independently. The detected (filtered) photon-number operators have the forms
\begin{gather}
\begin{aligned}
\hat{N}''_{\widetilde{\omega},a}&= \int^{\tilde{\omega}+\Delta\omega /2}_{\tilde{\omega}-\Delta\omega /2}\hspace{-3.5 mm}\mathrm{d}\omega\,\frac{1}{2}(\hat{a}_{\omega,z}'e^{i\phi}+\hat{a}_{\omega,s}')^{\dag}(\hat{a}_{\omega,z}'e^{i\phi}+\hat{a}_{\omega,s}')
\\&=\frac{1}{2}(\hat{N}_{\widetilde{\omega},s}'+\hat{N}_{\widetilde{\omega},z}')+\frac{1}{2}\int^{\tilde{\omega}+\Delta\omega /2}_{\tilde{\omega}-\Delta\omega /2}\hspace{-3.5 mm}\mathrm{d}\omega\,
(\hat{a}'^{\dag}_{\omega,z}{}\hat{a}_{\omega,s}'e^{-i\phi}+\hat{a}'^{\dag}_{\omega,s}\hat{a}_{\omega,z}'e^{i\phi}) \, ,
\end{aligned} \label{EqnEvolutionofNzetas''}
\\
\begin{aligned}
\hat{N}''_{\widetilde{\omega},b}&=\frac{1}{2}(\hat{N}_{\widetilde{\omega},s}'+\hat{N}_{\widetilde{\omega},z}')-\frac{1}{2}\int^{\tilde{\omega}+\Delta\omega /2}_{\tilde{\omega}-\Delta\omega /2}\hspace{-3.5 mm}\mathrm{d}\omega\,
(\hat{a}'^{\dag}_{\omega,z}\hat{a}_{\omega,s}'e^{-i\phi}+\hat{a}'^{\dag}_{\omega,s}\hat{a}_{\omega,z}'e^{i\phi}) \,. \label{EqnEvolutionofNzetaz''}
\end{aligned}
\end{gather}
Their sum, $\hat{N}''_{\widetilde{\omega},a}+\hat{N}''_{\widetilde{\omega},b}=(\hat{N}'_{\widetilde{\omega},s}+\hat{N}'_{\widetilde{\omega},z})$, gives the total detected photon number, while their difference gives the electro-optic signal
\begin{equation}
\hat{N}''_{\widetilde{\omega},a}-\hat{N}''_{\widetilde{\omega},b}=\int^{\tilde{\omega}+\Delta\omega /2}_{\tilde{\omega}-\Delta\omega /2}\hspace{-3.5 mm}\mathrm{d}\omega\,
(\hat{a}'^{\dag}_{\omega,z}\hat{a}_{\omega,s}'e^{-i\phi}+\hat{a}'^{\dag}_{\omega,s}\hat{a}_{\omega,z}'e^{i\phi}) \, .
\end{equation}
By utilizing the mean field approximation for the $z$-component of the electric field, these equations reduce to
\begin{gather}
\begin{aligned}
\hat{N}''_{\widetilde{\omega},a}+\hat{N}''_{\widetilde{\omega},b}&\approx\int^{\tilde{\omega}+\Delta\omega /2}_{\tilde{\omega}-\Delta\omega /2}\hspace{-3.5 mm}\mathrm{d}\omega\, |\alpha_z(\omega)|^2 \, ,
\end{aligned}
\\
\begin{aligned}
\hat{N}''_{\widetilde{\omega},a}-\hat{N}''_{\widetilde{\omega},b}&=\int^{\tilde{\omega}+\Delta\omega /2}_{\tilde{\omega}-\Delta\omega /2}\hspace{-3.5 mm}\mathrm{d}\omega\,
( \alpha^*_z(\omega) \hat{a}_{\omega,s}'e^{-i\phi}+\hat{a}'^{\dag}_{\omega,s} \alpha_z(\omega)e^{i\phi}) \, ,
\end{aligned}
\end{gather}
where we have defined $\alpha_z(\omega)|\alpha_z(\omega)\rangle =\hat{a}_{\omega,z}|\alpha_z(\omega)\rangle $. By rearranging the coordinates, we can set $t_0=0$, giving $\alpha_z(\omega)\approx \alpha_z(\widetilde{\omega})$ for a sufficiently small $\Delta \omega$. Considering Eq. \eqref{filter}, this approximation leads to
\begin{gather}
\begin{aligned}
\hat{N}''_{\widetilde{\omega},a}+\hat{N}''_{\widetilde{\omega},b}&\approx|\sqrt{\Delta\omega} \alpha_z(\widetilde{\omega})|^2 \, ,
\end{aligned}
\\
\begin{aligned}
\hat{N}''_{\widetilde{\omega},a}-\hat{N}''_{\widetilde{\omega},b}&\approx\sqrt{\Delta \omega}
\big(\hat{\mathfrak{u}}_{\widetilde{\omega},s}'\alpha^*_z(\widetilde{\omega})e^{-i\phi}+\hat{\mathfrak{u}}'^{\dag}_{\widetilde{\omega},s}\alpha_z(\widetilde{\omega})e^{i\phi}\big) \, .
\end{aligned}
\end{gather}
{By noting that the Minkowski vacuum state is unaffected by phase rotation, we can arbitrarily set the phase of the strong coherent signal, thus by setting ${\alpha}_z(\widetilde\omega)=|{\alpha}_z(\widetilde\omega)|$, we have:
\begin{align}
 \frac{\braket{\hat{N}''_{\widetilde{\omega},a}-\hat{N}''_{\widetilde{\omega},b}}}{\sqrt{\braket{\hat{N}''_{\widetilde{\omega},a}+\hat{N}''_{\widetilde{\omega},b}}}}&=\left\langle\frac{\alpha^*_{z}(\widetilde{\omega})}{|\alpha_{z}(\widetilde{\omega})|}\hat{\mathfrak{u}}'_{\widetilde{\omega}}e^{-i \phi}+\frac{\alpha_{z}(\widetilde{\omega})}{|\alpha_{z}(\widetilde{\omega})|}\hat{\mathfrak{u}}'^{\dag}_{\widetilde{\omega}}e^{i \phi} \right\rangle\, 
\\
&=\left\langle\hat{\mathfrak{u}}'_{\widetilde{\omega}}e^{-i \phi}+\hat{\mathfrak{u}}'^{\dag}_{\widetilde{\omega}}e^{i \phi} \right\rangle\, .
\end{align}
This gives an identical result to \eqref{EOquadrature}.}
\begin{figure}[t]\includegraphics[width=0.7\textwidth]{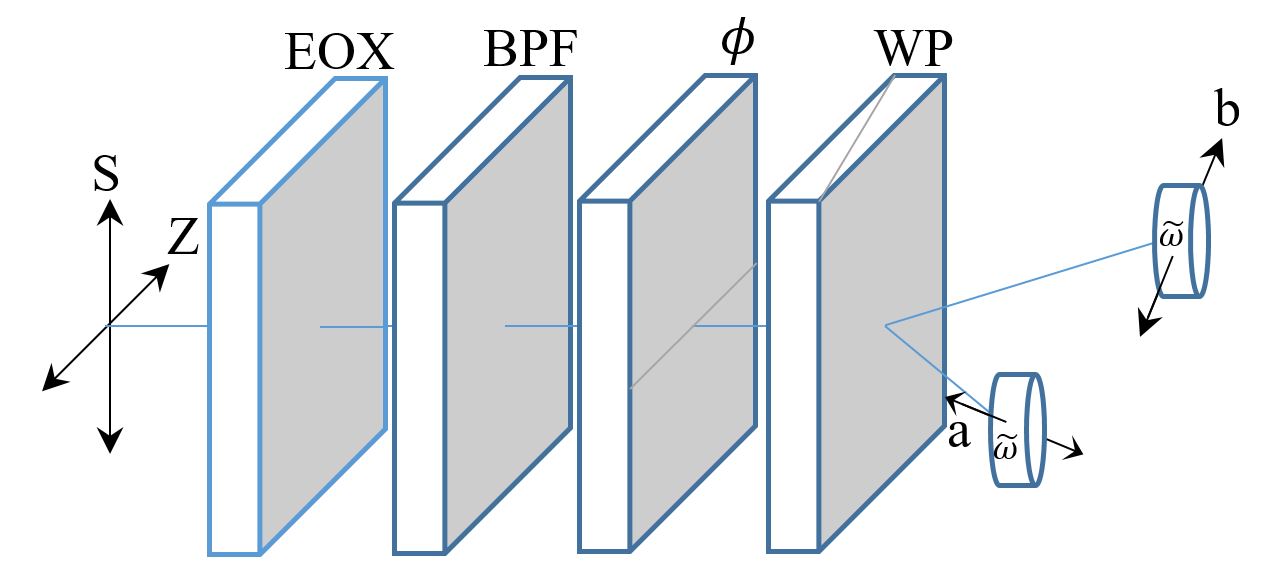}
\caption{The electric field passes through the $EOX$ for the signal to induce an interaction with the vacuum. The output goes through a band-pass filter, which filters all frequency except $\widetilde\omega+\eta/2\leq \omega \leq \widetilde\omega-\eta/2$). The $\phi$-waveplate in the z-polarization allows a homodyne detection of arbitrary phase, this is followed by a wollaston prism in the diagonal plane. The electric field is then detected with a photon counter for each polarisation.}
\label{FigExperimentalSetupHomodyne}
\end{figure}

\twocolumngrid
\bibliography{literature}

\end{document}